\documentclass[journal,draftclsnofoot,onecolumn,12pt,twoside]{IEEEtran}

\ifCLASSINFOpdf
\else
\fi
\usepackage[cmex10]{amsmath}
\usepackage{amsthm}
\usepackage{algorithm}
\usepackage{algorithmic}
\usepackage{amsmath}
\usepackage{amssymb}
\usepackage{amsfonts}
\usepackage{epsfig}

\newtheorem{defn}{\bf Definition}
\newtheorem{thm}{\bf Theorem}

\newtheorem{lem}[thm]{\bf Lemma}

\newcommand{\beq}{\begin{equation}}
\newcommand{\eeq}{\end{equation}}

\newcommand{\bea}{\begin{eqnarray}}
\newcommand{\eea}{\end{eqnarray}}
\newcommand{\bean}{\begin{eqnarray*}}
\newcommand{\eean}{\end{eqnarray*}}

\newcommand{\bit}{\begin{itemize}}
\newcommand{\eit}{\end{itemize}}

\newcommand{\ben}{\begin{enumerate}}
\newcommand{\een}{\end{enumerate}}

\newcommand {\nbf}[1]{\mbox{\boldmath $#1$} }
\newcommand {\snbf}[1]{\mbox{\scriptsize \boldmath $ #1$} }

\begin{document}
%
\title{Interleaved Polar (I-Polar) Codes}
%
%
%

\author{Mao-Ching Chiu,~\IEEEmembership{Member,~IEEE}
\thanks{
  This work was supported by the Ministry of Science and 
  Technology, Taiwan, R.O.C., under Grant MOST 
  	107-2221-E-194-003-MY2.}
\thanks{M.-C. Chiu is with the Department of
Communications Engineering, National Chung Cheng University,
Min-Hsiung, Chia-Yi, 621, Taiwan, R.O.C. (e-mail:
ieemcc@ccu.edu.tw). }
}

\maketitle



\begin{abstract}
By inserting interleavers between intermediate stages of the polar encoder, a new class of polar codes, termed interleaved polar (i-polar) codes, is proposed. By the uniform interleaver assumption, we derive the weight enumerating function (WEF) and input-output weight enumerating function (IOWEF) averaged over the ensemble of i-polar codes. The average WEF can be used to calculate the upper bound on the average block error rate (BLER) of a code selected at random from the ensemble of i-polar codes. Also, we propose a concatenated coding scheme that employs $P$ high rate codes as the outer code and $Q$ i-polar codes as the inner code with an interleaver in between. The average WEF of the concatenated code is derived based on the uniform interleaver assumption. Simulation results show that BLER upper bounds can well predict BLER performance levels of the concatenated codes. The results show that the performance of the proposed concatenated code with $P=Q=2$ is better than that of the CRC-aided i-polar code with $P=Q=1$ of the same length and code rate at high signal-to-noise ratios (SNRs). Moreover, the proposed concatenated code allows multiple decoders to operate in parallel, which can reduce the decoding latency and hence is suitable for ultra-reliable low-latency communications (URLLC).


\end{abstract}


\newpage

%
\IEEEpeerreviewmaketitle

\section{Introduction}
Polar codes \cite{Arikan2009} are constructed from the generator matrix $\nbf{G}_2^{\otimes M}$ with $\nbf{G}_2=\left[{1 \atop 1} {0 \atop 1}\right]$, where $\otimes M$ denotes the $M$th Kronecker power. It has been shown in \cite{Arikan2009}, that the synthesized channels seen by individual bits approach two extremes, either a noiseless channel or a pure-noise channel, as the block length $N = 2^M$ grows large. The fraction of noiseless channels is close to the channel capacity. Therefore, the noiseless channels, termed unfrozen bit channels, are selected for transmitting message bits while the other channels, termed frozen bit channels, are set to fixed values known by both encoder and decoder. Therefore, polar codes are the first family of codes that achieve the capacity of symmetric binary-input discrete memoryless channels under a low-complexity successive cancellation (SC) decoding algorithm as the block length $N$ approaches infinity.

However, the performance of polar codes at short to moderate block lengths is disappointing under the SC decoding algorithm.
Later, a successive cancellation list (SCL) decoding algorithm for polar codes was proposed \cite{Tal2015}, which approaches the performance of the maximum-likelihood (ML) decoder as the list size $L$ is large. However, the performance levels of polar codes are still inferior to those of low-density parity-check (LDPC) codes even under the ML decoder. To strengthen polar codes, a serial concatenation of a cyclic redundancy check (CRC) code and a polar code, termed the CRC-aided polar code, was found to be effective to improve the performance under the SCL decoding algorithm \cite{Tal2015}. The performance levels of CRC-aided polar codes under the SCL decoding algorithm are better than those of LDPC and turbo codes \cite{Tal2015, Li2012}.

As the SCL decoder is capable to achieve the ML performance, it is important to study the block error rate (BLER) of polar codes under the ML decoder. However, in the literature, there are no analytical results regarding the ML performance of polar codes. The BLERs of polar codes rely on simulations that are time-consuming. A possible way to analyze the BLER performance of a coding scheme is to use the BLER upper bound which is a function of the weight enumerating function (WEF) as that used to analyze turbo codes \cite{Benedetto1996}. 
However, if the code size is large, obtaining the exact WEF of a polar code with the heuristic method is prohibitively complex. Approximations of WEFs of polar codes are proposed in \cite{Valipour2013, Zhang2017} based on the probabilistic weight distribution (PWD) \cite{Hirotomo2005}. 

In this paper, we propose to randomize the polar code using interleavers between the intermediate stages of the polar code encoder. Codes constructed on the basis of this idea are called interleaved polar (i-polar) codes. The ensemble of i-polar codes is formed by considering all possible interleavers. The regular polar code is just one realization of the ensemble of i-polar codes. Based on the concept of uniform interleaver, i.e., all interleavers are selected uniformly at random from all possible permutations, the average WEF of a code selected at random from the ensemble of i-polar codes can be evaluated. The concept of uniform interleaver has also been used in the analysis of turbo codes  \cite{Benedetto1996}. Note that the WEF analysis in this paper is not an approximation to the WEF of a polar code, but is an exact WEF averaged over the ensemble of i-polar codes. Based on the average WEF, a BLER upper bound, termed {\em simple bound} \cite{Divsalar1999}, can be used to evaluate the BLER performance averaged over the ensemble of codes. Simulation results show that the BLER upper bounds can well predict the ML  performance levels of i-polar codes at high SNRs. 
Also, we will show by simulations that a specific realization of i-polar codes outperforms a regular polar code under the SCL decoder of the same list size. 

We also propose a concatenated coding scheme that employs $P$ identical high rate codes as the outer code and $Q$ identical i-polar codes as the inner code with an interleaver in between. CRC codes are the most popular outer codes employed in the concatenation of polar codes. We propose as an alternative to use systematic regular repeat-accumulate (RRA) codes or irregular repeat-accumulate (IRA) codes \cite{Jin2000} as the outer component code. The average WEF of the concatenated code is derived based on the uniform interleaver assumption. Simulation results show that the BLER upper bounds can well predict the BLER performance levels of the concatenated codes. One advantage of the proposed concatenated code is that, for $Q > 1$, the code can be decoded using $Q$ SCL decoders working in parallel which can significantly reduce the decoding latency when $Q$ is large. Analytical and simulation results both show that the performance of the proposed concatenated code with $P=Q=2$ is better than that of the CRC-aided i-polar code with $P=Q=1$ of the same length and code rate at high SNRs. Therefore, the proposed coding scheme is suitable for ultra-reliable low-latency communications (URLLC) \cite{ITU-R2017}.

The rest of the paper is organized as follows. We begin with a brief introduction of polar codes in Section \ref{sec:back}. The construction of i-polar codes is presented in Section \ref{sec:IPC}. Section \ref{sec:WEF} presents the WEF and IOWEF analysis of i-polar codes. In Section \ref{sec:concat}, a concatenated coding scheme with the i-polar code as the inner component code is proposed and the WEF of the concatenated code is presented. Analytical and simulation results are given in Section \ref{sec:simu}. Finally, conclusions are given in Section \ref{sec:conc}.

Notations: Throughout this paper, matrices and vectors are set in boldface, with upper case letters for matrices and lower case letters for vectors. An $N$-tuple vector $\nbf{x}$ is denoted as $\nbf{x} = [x_0, x_1, \ldots, x_{N-1}]$ with the indices starting from 0 (instead of 1 for normal vector representations). The notation $\nbf{x}_{a}^{b}$ means the sub-vector $[x_a, x_{a+1}, \ldots, x_{b}]$ if $b \geq a$ and null vector otherwise. 
Set quantities such as ${\cal A}$ are denoted using the calligraphic font, and the cardinality of the set ${\cal A}$ is denoted as $|{\cal A}|$. 

\section{Background}
\label{sec:back}
A codeword of the polar code of length $N = 2^M$ without bit-reversal matrix can be represented by
\beq
      \nbf{x} = \nbf{u} \nbf{G}_2^{\otimes M}, \label{eq:genm}
\eeq
where $\nbf{u} = [u_0, u_1, \ldots, u_{N-1}]$ is the message bits, $\nbf{x} = [x_0, x_1, \ldots, x_{N-1}]$ is the codeword bits, and $\nbf{G}_2=\left[{1 \atop 1} {0 \atop 1}\right]$. A polar code of block length $2^M$ can be represented by a graph with $M$ layers of trellis connections as given in \cite{Arikan2009} which is called the standard graph of the polar code. It has been shown in \cite{Hussami2009a} that for a polar code of block length $2^M$, there exist $M!$ different
graphs obtained by different permutations of the $M$ layers of trellis connections. We consider to represent a polar code with reverse ordering of its standard graph. Figure \ref{fig:polar_enc8} shows an example graph for $N = 8$ with reverse ordering of the standard graph, where the notation $\oplus$ represents a modulo-2 adder.
\begin{figure}[!t]
\centering
\includegraphics[width=0.8\columnwidth]{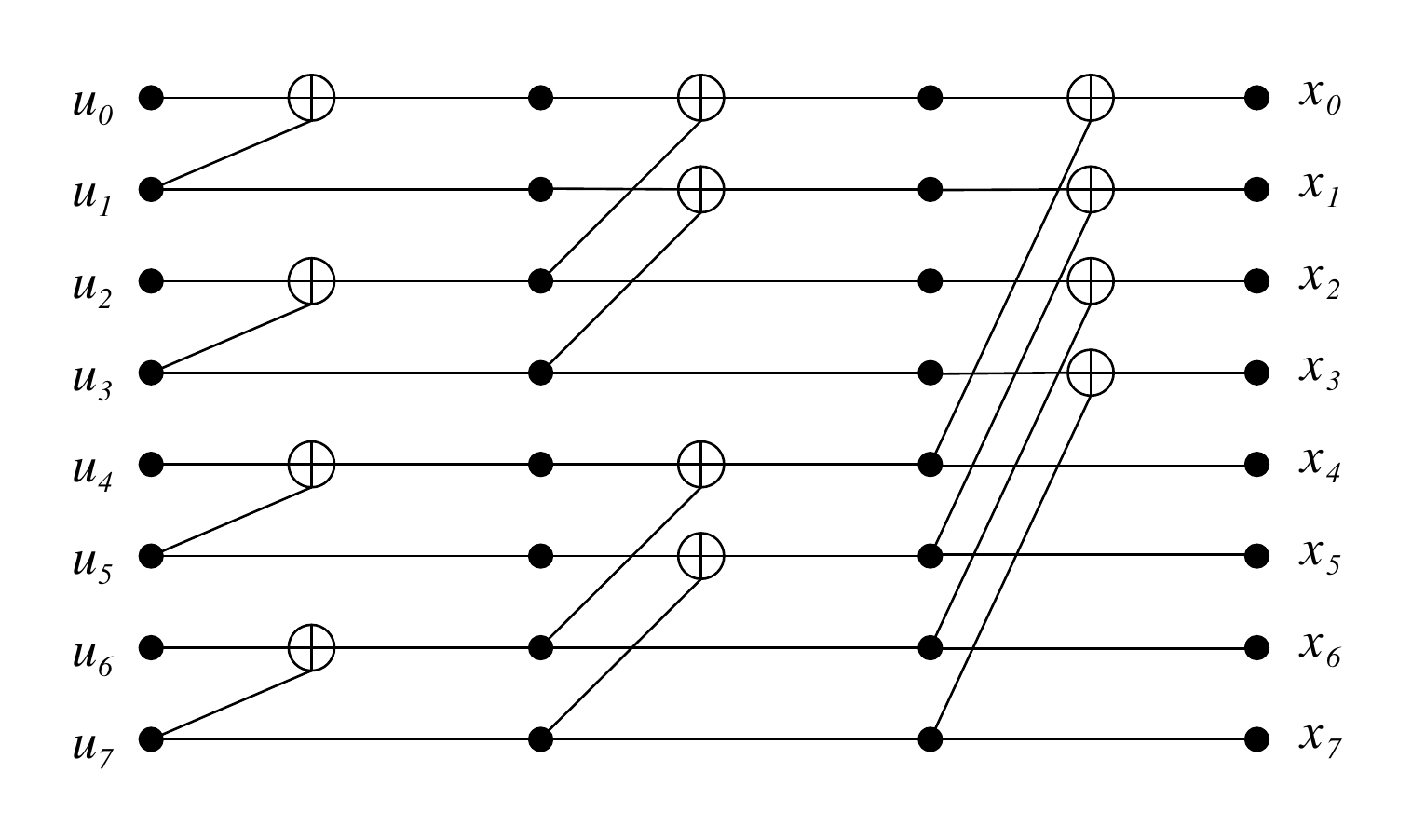}
\caption{Graph representation of a polar code of block length $N=8$.}
\label{fig:polar_enc8}
\end{figure}

The codeword $\nbf{x}$ obtained from (\ref{eq:genm}) is then transmitted via $N$ independent uses of the binary input discrete memoryless channel (B-DMC) $W:{\cal X} \rightarrow {\cal Y}$, where ${\cal X} = \{0, 1\}$ denotes the input alphabet, ${\cal Y}$ denotes the output alphabet, and $W(y|x)$ denotes the channel transition probabilities. The conditional distribution of the output ${\nbf Y} = {\nbf{y}}$ given the input ${\nbf{X}} = \nbf{x}$, denoted as $W^N(\nbf{y}|\nbf{x})$, is given by
\[
W^N(\nbf{y}|\nbf{x}) = \prod_{i=0}^{N-1} W(y_i|x_i).
\]
The distribution of $\nbf{Y}$ conditioned on $\nbf{U}=\nbf{u}$, denoted as $W_M(\nbf{y}|\nbf{u})$, is given by
\[
     W_M(\nbf{y}|\nbf{u}) = W^N(\nbf{y}|\nbf{u}\nbf{G}_2^{\otimes n}).
\]
The polar code of length $N=2^M$ transfers the original $2^M$ identical channels $W$ into $2^M$ synthesized channels, denoted as $W_M^{(i)}: {\cal X} \rightarrow {\cal Y}^N \times {\cal X}^i$ for $i \in \{0, \ldots, 2^M-1\}$ with the transition probability given by
\[
    W_M^{(i)}(\nbf{y}, \nbf{u}_0^{i-1}|u_i) \equiv \sum_{\nbf{u}_{i+1}^{N-1} \in {\cal X}^{N-i-1}} \frac{1}{2^{N-1}}W_M(\nbf{y}|\nbf{u}).
\]
It has been shown in \cite{Arikan2009} that as $M$ grows large, the synthesized channels start polarizing. They approach either a noiseless channel or a pure-noise channel. The fraction of noiseless channels is close to the channel capacity. Therefore, the noiseless channels are selected for transmitting message bits while the other channels are set to fixed values known by both encoder and decoder. In the code design, a polar code of dimension $K$ is generated by selecting the $K$ least noisy channels among $W_M^{(i)}$ and the indices of the $K$ least noisy channels are denoted as a set ${\cal A}$. Define $\nbf{u}_{\cal A}$ as a sub-vector of $\nbf{u}$ formed by the elements of $\nbf{u}$ with indices in ${\cal A}$. Only the sub-vector $\nbf{u}_A$, termed {\em unfrozen bits}, is employed to transmit message bits. The other bits $\nbf{u}_{{\cal A}^c}$, termed {\em frozen bits}, are set to fixed values known by both encoder and decoder. In this paper, we set the frozen bits to all zeros.

Polar codes can be decoded by the SC decoder which has decoding complexity of $O(N\log N)$ and can achieve the capacity as $N$ approaches infinity \cite{Arikan2009}. However, SC decoder does not perform well at short to moderate block lengths. The SC decoder has the drawback that if a bit is not correctly detected, it is not possible to correct it in future decoding steps. To improve the performance, a more sophisticated SCL decoder was proposed in \cite{Tal2015}, which performs very close to the ML performance for large list size $L$. 
The SCL decoder of list size $L$ is based on the tree search over the message bits under the complexity constraint that the number of candidates in the list is at most $L$. 
At the $i$th step, if $i \in {\cal A}$, the decoder extends every candidate path in the list along two paths of the binary tree by appending a bit 0 or a bit 1 to each of the candidate paths.
Therefore, for every $i \in {\cal A}$, the decoder doubles the number of paths. When the number of paths exceeds $L$, only $L$ most reliable paths are retained. This procedure is repeated until $i=N-1$. At the last step, the most reliable path is selected as the output of the decoder. The SCL decoder degenerates to the SC decoder when $L = 1$. The details of the SCL decoder can be found in \cite{Tal2015} based on the probability domain and in \cite{Balatsoukas-Stimming2015} based on the log-likelihood ratio (LLR) domain.

\begin{figure}[!t]
\centering
r\includegraphics[width=0.8\columnwidth]{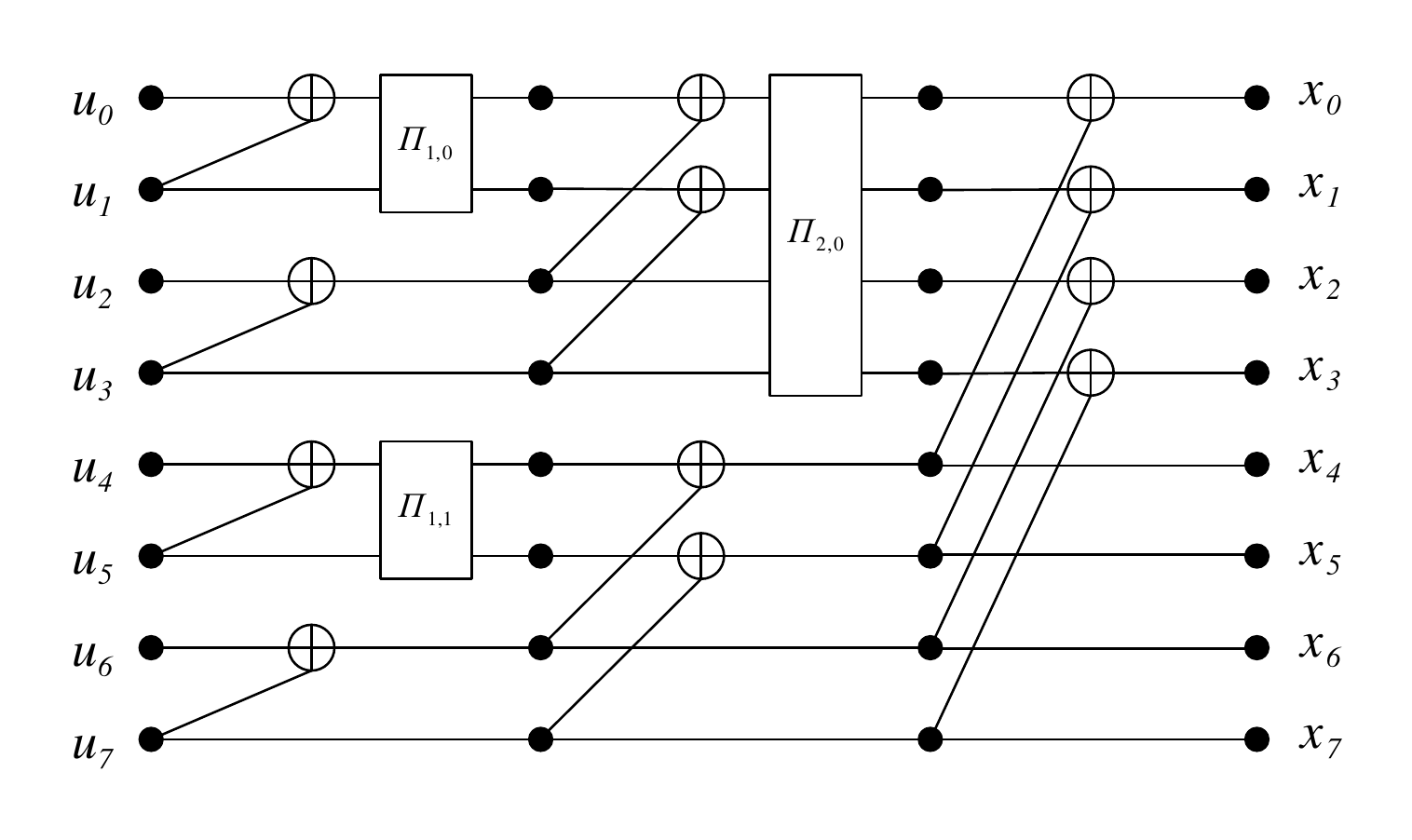}
\caption{Graph representation of an interleaved polar code of block length $N=8$.}
\label{fig:rand_polar_enc8}
\end{figure}

\section{Interleaved Polar (I-Polar) Codes}
\label{sec:IPC}
A polar code is constructed recursively by the well-known structure $\nbf{x}+\nbf{y}|\nbf{y}$. Without ambiguity, we use '+' to denote the binary addition as well as the ordinary real number addition. Let ${\cal C}_1$ and ${\cal C}_2$ be two linear codes of the same length $n$. We define ${\cal C}_1+{\cal C}_2|{\cal C}_2$ as
\[
     {\cal C}_1+{\cal C}_2|{\cal C}_2 = \left\{[\nbf{x}+\nbf{y},\nbf{y}]| \nbf{x} \in {\cal C}_1, \nbf{y} \in{\cal C}_2\right\}.
\]
As shown by the graph representation of the polar code, a polar code can be described by the following recursive equation
\[
     {\cal C}_{m,j} = {\cal C}_{m-1,2j} + {\cal C}_{m-1,2j+1}|{\cal C}_{m-1,2j+1},
\]
for $m=1, \ldots, M$ and $j = 0, \ldots, 2^{M-m}-1$. The initial conditions are ${\cal C}_{0,j} = \{0, 1\}$ if $j \in {\cal A}$ and ${\cal C}_{0,j} = \{0\}$ if $j \in {\cal A}^c$. The polar code of length $N=2^M$ is represented by the code ${\cal C}_{M,0}$.

We propose to construct the i-polar code by inserting an interleaver at the output of every upper encoder for ${\cal C}_{m-1,2j}$. An interleaver can be represented as a permutation matrix $\nbf{\Pi}$. We define $ {\cal C} \nbf{\Pi}$ as 
\[
   {\cal C}  \nbf{\Pi} = \{\nbf{x} \nbf{\Pi}| \nbf{x} \in {\cal C}\}
\] 
which represents a code obtained by permuting the code bits of all codewords of ${\cal C}$ using the interleaver $\nbf{\Pi}$. Therefore, the i-polar code can be described by the following recursive equation
\beq
    {\cal C}_{m,j} ={\cal C}_{m-1,2j} \nbf{\Pi}_{m-1,j} + {\cal C}_{m-1,2j+1}|{\cal C}_{m-1,2j+1},
\label{eq:rpr}
\eeq
for $m=1, \ldots, M$ and $j = 0, \ldots, 2^{M-m}-1$. The initial conditions are ${\cal C}_{0,j} = \{0, 1\}$ if $j \in {\cal A}$ and ${\cal C}_{0,j} = \{0\}$ if $j \in {\cal A}^c$. Note that the interleavers $\nbf{\Pi}_{0,j} = 1$, for $j=0, \ldots, 2^{M-1}-1$, are trivial. At the $m$th layer, there are $2^{M-m-1}$ interleavers of size $2^{m}$. Figure \ref{fig:rand_polar_enc8} shows the graph of an i-polar code of length $N = 8$, for which three interleavers are required, i.e., $\nbf{\Pi}_{1,0}$, $\nbf{\Pi}_{1,1}$, and $\nbf{\Pi}_{2,0}$ with sizes 2, 2, and 4, respectively. Note that the interleavers $\nbf{\Pi}_{0,j}$, for $j=0,1,2,3$, are omitted because they are trivial.

\begin{figure}[ht]
\centering
\includegraphics[width=0.8\columnwidth]{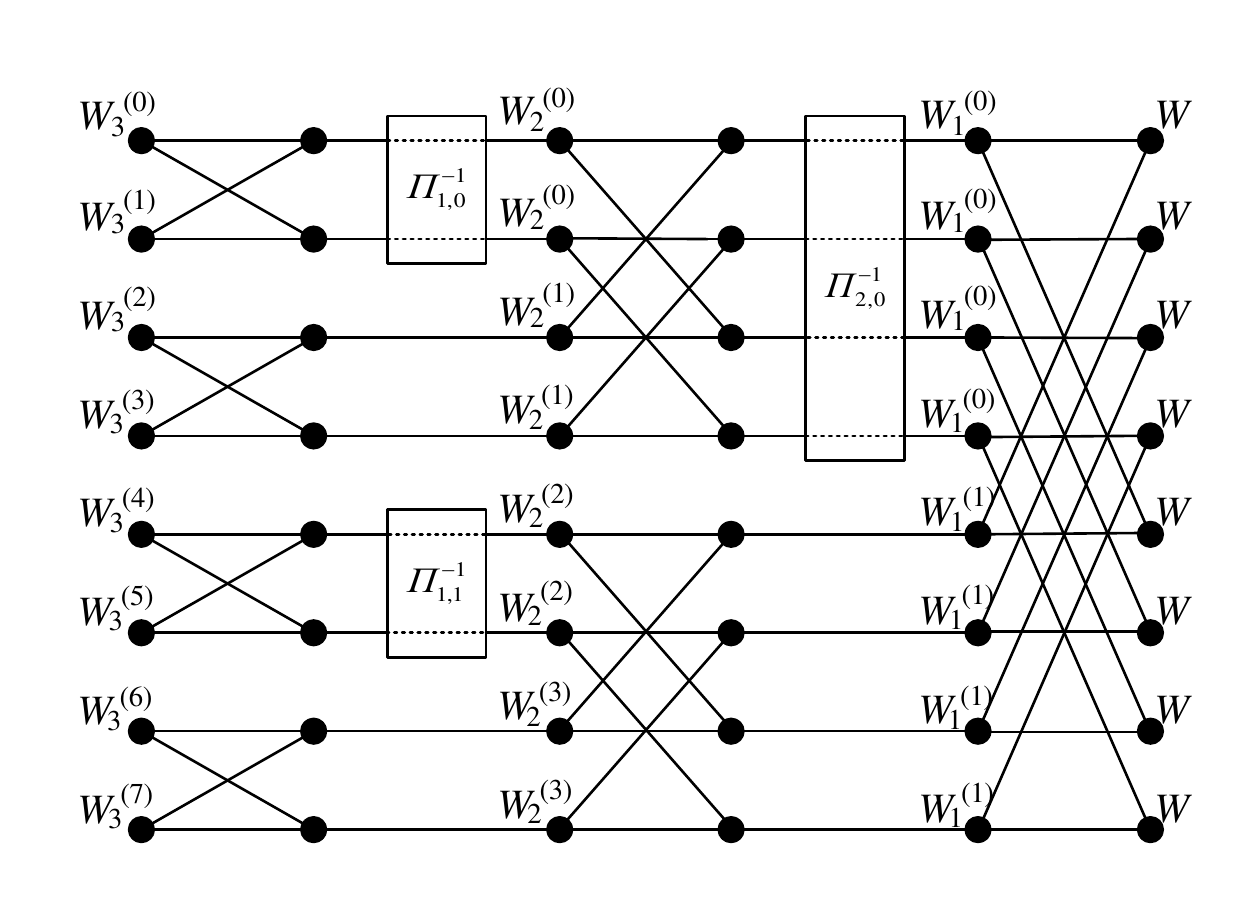}
\caption{The channel transformation process with $N=8$ channels.}
\label{fig:rand_polar_ch_tr}
\end{figure}

The following theorem shows that the interleavers do not change the polarization effect.
\begin{thm}
\label{th:thm0}
I-polar codes have the same polarization effect as polar codes.
\end{thm}
\begin{proof}
Reading from right to left, Figure \ref{fig:rand_polar_ch_tr} illustrates the channel transformation process of the i-polar code for $N=8$ with synthesized channels de-interleaved by $\nbf{\Pi}_{1,0}^{-1}$, $\nbf{\Pi}_{1,1}^{-1}$, and $\nbf{\Pi}_{2,0}^{-1}$. The channel transformation process of the polar code is formed by replacing the de-interleavers with direct links  represented by dashed lines in Figure \ref{fig:rand_polar_ch_tr}. For channel transformation process of the polar code, the figure starts with $2^{M-1}$ copies of the transformation $(W,W)\longmapsto (W_1^{(0)}, W_1^{(1)})$. The transformation continues in butterfly patterns with $2^{M-\mu-1}$ copies of the transformation $(W_\mu^{(i)},W_\mu^{(i)})\longmapsto (W_{\mu+1}^{(2i)}, W_{\mu+1}^{(2i+1)})$ for $i=0, \ldots, 2^{\mu}-1$ and $\mu = 0, \ldots, M-1$ with $W_0^{(0)} = W$. Finally, the synthesized channels $W_{M}^{(i)}$ for $i=0, \ldots, 2^M-1$ can be obtained. The synthesized channels at the intermediate stages can be represented as a binary tree similar to that shown in Figure 6 of \cite{Arikan2009}. The root node represents the channel $W$. The root $W$ gives birth to an upper channel $W_1^{(0)}$ and a lower channel $W_1^{(1)}$, which are represented by the two nodes at level 1. The channel $W_1^{(0)}$ in turn gives birth to channels $W_2^{(0)}$ and  $W_2^{(1)}$, and the channel $W_1^{(1)}$ gives birth to channels $W_2^{(2)}$ and  $W_2^{(3)}$, and so on. Then based on the concept of a random tree process, the polarization effect is proved in \cite{Arikan2009}. We want to prove that inserting de-interleavers at the upper channels as that shown in Figure \ref{fig:rand_polar_ch_tr} does not change the polarization effect. Note that there are $2^{M-\mu-1}$ copies of the transformation $(W_\mu^{(i)},W_\mu^{(i)})\longmapsto (W_{\mu+1}^{(2i)}, W_{\mu+1}^{(2i+1)})$. For the i-polar code, after channel transformation, the   $2^{M-\mu-1}$ copies of the upper channels $W_{\mu+1}^{(2i)}$ are de-interleaved by $\nbf{\Pi}_{M-\mu-1, i}^{-1}$. Since the de-interleaver acts only on the $2^{M-\mu-1}$ channels of the same type $W_{\mu+1}^{(2i)}$, the outputs of the de-interleaver are just re-ordered channels of the same type $W_{\mu+1}^{(2i)}$ which is the same as that of the original polar code. Therefore, by induction, for the i-polar code, further transformation of the $2^{M-\mu-1}$ re-ordered channels of the same type gets the same synthesized channels as those of the polar code.  
\end{proof}

The SC or SCL decoder for polar codes can be easily modified to decode i-polar codes. 
As proved in Theorem \ref{th:thm0}, the i-polar code and polar code produce the same synthesized channels $W_M^{(i)}$ for $i=0, \ldots, 2^M-1$. 
Therefore, the same bit channel selection algorithm 
as those designed for polar codes can be employed for i-polar codes. 
It has been shown that the bit channel selection algorithms based on Gaussian approximation (GA) for density evolution such as those proposed in \cite{Trifonov2012, Chiu2013, Dai2017} are effective for binary-input additive white Gaussian noise (BI-AWGN) channels. Bhattacharyya parameter can calso be employed for bit channel selection \cite{Zhao2011}. In this paper, we employ the bit channel selection algorithm given in \cite{Chiu2013} for both i-polar and polar codes. 
For convenience, we give a brief review of the algorithm proposed in \cite{Chiu2013}. Assume that the all-zero codeword was transmitted. The bit LLR under SC decoding for the  channel $W_\mu^{(i)}$ is defined as $L_\mu^{(i)} = \log(W_\mu^{(i)}(\cdot|0)/W_\mu^{(i)}(\cdot|1))$. The idea of GA is to approximate the LLR as a Gaussian random variable with mean $\beta$ and variance $\sigma^2$ satisfying $\sigma^2 = 2\beta$. Therefore, the p.d.f.~of the LLR random variable can be described by a single parameter $\sigma$. The mutual information of the channel $W_\mu^{(i)}$, defined as $I_\mu^{(i)} = I(W_\mu^{(i)})$, was shown in \cite{Brink2001} to be 
\[
J(\sigma) = 1- \int_{-\infty}^{\infty} \frac{e^{-(x-\sigma^2/2)^2/(2\sigma^2)}}{\sqrt{2 \pi} \sigma} \log_2(1+e^{-x})dx,
\] 
where $\sigma^2$ is the variance of the LLR random variable $L_\mu^{(i)}$. We want to find the transformation of mutual information $(I_\mu^{(i)},I_\mu^{(i)})\longmapsto (I_{\mu+1}^{(2i)}, I_{\mu+1}^{(2i+1)})$ that corresponds to the mutual information for the channel transformation $(W_\mu^{(i)},W_\mu^{(i)})\longmapsto (W_{\mu+1}^{(2i)}, W_{\mu+1}^{(2i+1)})$. The initial condition  is $I_0^{(0)}=J(2/\sigma_n)$, where $\sigma_n^2=N_0/2$ is the noise variance of the AWGN channel. Now under SC decoding, assuming that the upper branch is correctly decoded, $L_{\mu+1}^{(2i+1)}$ is the sum of two i.i.d. Gaussian random variables with variance $[J^{-1}(I_\mu^{(i)})]^2$. Therefore  $L_{\mu+1}^{(2i+1)}$ is a Gaussian random variable with variance $2[J^{-1}(I_\mu^{(i)})]^2$, and hence the mutual
information $I_{\mu+1}^{(2i+1)}$ is given by
\beq
       I_{\mu+1}^{(2i+1)} = J(\sqrt{2} J^{-1}(I_\mu^{(i)})),
       \label{eq:im1}
\eeq
for $i = 0, \ldots, 2^\mu-1$.
Also, according to Proposition 4 of \cite{Arikan2009}, $I_{\mu+1}^{(2i)} + I_{\mu+1}^{(2i+1)} = 2I_\mu^{(i)}$,  and hence
\beq
    I_{\mu+1}^{(2i)} = 2I_{\mu}^{(i)} - J(\sqrt{2} J^{-1}(I_\mu^{(i)})),
    \label{eq:im2}
\eeq
for $i = 0, \ldots, 2^\mu-1$. Through the recursions of (\ref{eq:im1}) and (\ref{eq:im2}) with $\mu=0, \ldots, M-1$, we can calculate $I_{M}^{(i)}$ for $i = 0, \ldots, 2^M-1$. The subset ${\cal A}$ is selected such that if $i \in {\cal A}$ then $I_M^{(i)} > I_M^{(j)}$ for all $j \in {\cal A}^c$.

Given the same set ${\cal A}$, the i-polar code has the same performance level as that of the polar code under the SC decoder, since the synthesized channels, $W_M^{(i)}$ for $i \in {\cal A}$, are the same for both codes as shown in the proof of Theorem \ref{th:thm0}.
However, they have different performance levels when a more sophisticated decoder, such as the SCL decoder \cite{Tal2015} or stack decoder \cite{Niu2012, Niu2012-2}, is employed.
We will show that the WEF of the i-polar code is different from that of the polar code. This implies that these two codes have different performance levels under the ML decoder.  Actually, simulation results show that i-polar codes perform  better than polar codes under the SCL decoder. 


\section{WEF and IOWEF of I-Polar Codes}
\label{sec:WEF}
\subsection{$(N,K,{\cal A})$ Ensemble of I-Polar Codes}
As described in Section \ref{sec:IPC}, at the $m$th layer of the i-polar graph, there are $2^{M-m-1}$ interleavers of size $2^{m}$. For an interleaver of size $n$, there are $n!$ possible interleavers. The $(N,K,{\cal A})$ ensemble of i-polar codes is formed by all possible interleavers given the unfrozen bit set ${\cal A}$, of which the code length $N=2^M$ and dimension $K=|\cal {\cal A}|$. In theory, it is impossible to exhaustively enumerate the WEF of i-polar codes over all possible interleavers when the code size is large. To overcome this difficulty, we assume that all interleavers are selected independently at random and each interleaver follows the uniform assumption as that used in the analysis for turbo codes \cite{Benedetto1996}.

\begin{defn} \cite{Benedetto1996} 
A uniform interleaver of length $n$ is a probability device selected in random over all possible interleavers which maps a given input binary vector of weight $d$ to all $\left(n \atop  d\right)$ permutations with equal probability $1/\left(n \atop  d\right)$.
\end{defn}

\subsection{WEF and IOWEF}

\begin{figure}[!t]
\centering
\includegraphics[width=1.0\columnwidth]{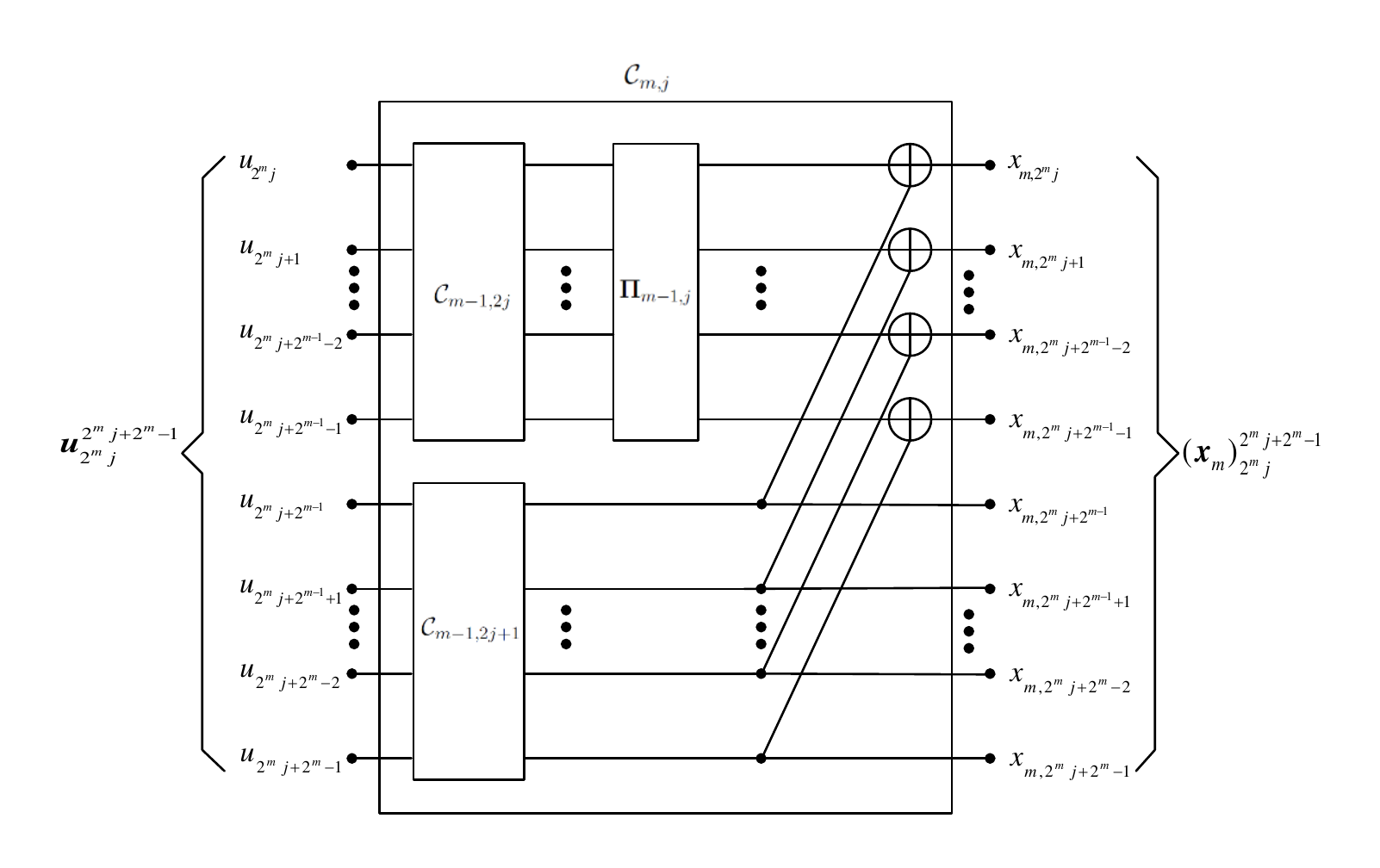}
\caption{Construction of ${\cal C}_{m, j}$ from ${\cal C}_{m-1, 2j}$ and ${\cal C}_{m-1, 2j+1}$. The input is $\nbf{u}_{2^mj}^{2^mj+2^m-1}$ and the output is $(\nbf{x}_m)_{2^mj}^{2^mj+2^m-1}$, where $\nbf{x}_m$ is the output at the $m$th stage of the i-polar encoder.} 
\label{fig:intp}
\end{figure}

Given an $(N, K)$ linear block code $\cal C$, its WEF is defined as
\bean
    A^{\cal C}(Y)&=&\sum_{\snbf{c} \in {\cal C}} Y^{w_H(\snbf{c})}\\
    &=&\sum_{d=0}^n A^{\cal C}_d Y^d,
\eean
where $w_H(\nbf{c})$ is the Hamming weight of $\nbf{c}$, $A^{\cal C}_d$ is the number of codewords of ${\cal C}$ with Hamming weight $d$, and $Y$ is a dummy variable. The WEF can be used to compute the exact probability of undetected errors and an upper bound on the BLER \cite{Divsalar1999}.

We define the input-output weight enumerating function (IOWEF) of the code ${\cal C}$ as 
\[
    A^{\cal C}(X,Y) = \sum_{w,d} A^{\cal C}_{w,d} X^w Y^d,
\]
where $A^{\cal C}_{w,d}$ denotes the number of codewords of ${\cal C}$ generated by an input message word of Hamming weight $w$ whose output codeword has Hamming weight $d$. It should be noted that the WEF is a polynomial in one variable and the IOWEF is a polynomial in two variables. The relation between the WEF and IOWEF is given by
\[
   A^{\cal C}(Y) = A^{\cal C}(X = 1, Y).
\]  
The WEF depends only on the code ${\cal C}$, as only weights of codewords are enumerated. However, the IOWEF depends on the encoder, as it depends on the pairs of Hamming weights of input message word and output codeword. Since there are many different encoders that generate the same code ${\cal C}$, we will assume that a specific encoder is employed when the IOWEF is considered.
The IOWEF can be used to compute the upper bound on the bit error rate (BER)  \cite{Divsalar1999}. Also it is important for the study of concatenated coding schemes. 

Due to the recursive equation (\ref{eq:rpr}), the code ${\cal C}_{m,j}$ can be represented as the graph shown in Figure \ref{fig:intp}. The code ${\cal C}_{m,j}$ forms a $(2^m, |{\cal A}_{2^mj}^{2^mj+2^m-1}|, {\cal A}_{2^mj}^{2^mj+2^m-1}-2^mj )$ ensemble of i-polar codes, where ${\cal A}_{a}^{b} \triangleq \{a,a+1, \ldots, b\} \cap {\cal A}$ and ${\cal A}-b \triangleq \{a-b: \forall a \in {\cal A}\}$.
The input of the encoder of ${\cal C}_{m,j}$ is $\nbf{u}_{2^mj}^{2^mj+2^m-1}$ and the output is $(\nbf{x}_m)_{2^mj}^{2^mj+2^m-1}$, where $\nbf{x}_m$ is the output vector at the $m$-th stage of the $i$-polar encoder.
Define the  WEF of the i-polar code ${\cal C}_{m,j}$ averaged over the ensemble as $A^{{\cal C}_{m,j}}(Y)$. 
Assume that the WEFs $A^{{\cal C}_{m-1,2j}}(Y)$ and $A^{{\cal C}_{m-1,2j+1}}(Y)$ are known. Then $A^{{\cal C}_{m,j}}(Y)$ is a function of $A^{{\cal C}_{m-1,2j}}(Y)$ and $A^{{\cal C}_{m-1,2j+1}}(Y)$. The following lemma is important to calculate $A^{{\cal C}_{m,j}}(Y)$.
\begin{lem}
\label{th:lem1}
Let $\nbf{c}_1$ and $\nbf{c}_2$ be length-$n$ binary vectors with Hamming weights $d_1$ and $d_2$, respectively. Assume that $\nbf{\Pi}$ is a uniform interleaver. 
Then the weight distribution for $\nbf{c}_1 \nbf{\Pi} + \nbf{c}_2|\nbf{c}_2$
averaged over all possible interleavers $\nbf{\Pi}$ is given by
\bean
     E_{\snbf{\Pi}}\left[Y^{w_H( \snbf{c}_1 \snbf{\Pi} + \snbf{c}_2|\snbf{c}_2)}\right] & \triangleq & \sum_{\snbf{\Pi}} P\{\nbf{\Pi}\} Y^{w_H( \snbf{c}_1 \snbf{\Pi} + \snbf{c}_2|\snbf{c}_2)} \\
&=&\sum_{k = \max(0, d_1+d_2 - n)}^{\min(d_1,d_2)} \frac{\left(d_2 \atop k\right)\left(n-d_2 \atop d_1-k\right)}{\left(n \atop d_1\right)}Y^{d_1+2d_2-2k}.
\eean
\end{lem}
\begin{proof}
We first derive the weight distribution for $ \nbf{c}_1 \nbf{\Pi} + \nbf{c}_2$. 
Since the weight of $\nbf{c}_1$ is $d_1$, there are $\left(n \atop d_1\right)$ permutations of $\nbf{c}_1\nbf{\Pi}$ with equal probability $1/\left(n \atop d_1\right)$. Among these permutations, let $k$ be the number of positions at which the elements in $ \nbf{c}_1\nbf{\Pi}$ and $\nbf{c}_2$ are both equal to 1. The minimum value of $k$ can be easily shown to be $ \max(0, d_1+d_2 - n)$ and the maximum value of $k$ to be $\min(d_1,d_2)$. Given the value $k$, the Hamming weight of $ \nbf{c}_1 \nbf{\Pi} + \nbf{c}_2$ is $d_1+d_2-2k$, and there are a total of $\left(d_2 \atop k\right)\left(n-d_2 \atop d_1-k\right)$ such permutations. Therefore, the probability that  $\nbf{c}_1 \nbf{\Pi}  + \nbf{c}_2$ has weight $d_1+d_2-2k$ is given by $\left(d_2 \atop k\right)\left(n-d_2 \atop d_1-k\right)/\left(n \atop d_1\right)$. Finally, the additional concatenation of $\nbf{c}_2$ gives the Hamming weight of $\nbf{c}_1 \nbf{\Pi} + \nbf{c}_2|\nbf{c}_2$ as $d_1+d_2-2k+d_2=d_1+2d_2-2k$. The proof is completed. 
\end{proof}
We are ready to calculate the average WEF $A^{{\cal C}_{m,j}}(Y)$ based on the recursive equation (\ref{eq:rpr}). 

\begin{thm}
\label{lem:awef}
Given the WEFs $A^{{\cal C}_{m-1,2j}}(Y)$ and $A^{{\cal C}_{m-1,2j+1}}(Y)$, the  WEF $A^{{\cal C}_{m,j}}(Y)$ of the code ${\cal C}_{m,j} ={\cal C}_{m-1,2j} \nbf{\Pi}_{m-1,j} + {\cal C}_{m-1,2j+1}|{\cal C}_{m-1,2j+1}$ averaged over all possible interleavers $\nbf{\Pi}_{m-1,j}$ is 
\bean
     A^{{\cal C}_{m,j}}(Y)  
     &=& \sum_{d_1, d_2}A^{{\cal C}_{m-1, 2j}}_{d_1} A^{{\cal C}_{m-1, 2j+1}}_{d_2} \sum_{k = \max(0, d_1+d_2 - n)}^{\min(d_1,d_2)} \frac{\left(d_2 \atop k\right)\left(n-d_2 \atop d_1-k\right)}{\left(n \atop d_1\right)} Y^{d_1+2d_2-2k} \\
     &\triangleq& H_m(A^{{\cal C}_{m-1, 2j}}(Y), A^{{\cal C}_{m-1, 2j+1}}(Y)), 
\eean
where $n = 2^{m-1}$.
\end{thm}
\begin{proof}
The averaged WEF of ${\cal C}_{m,j}$ can be written as
\bean
     A^{{\cal C}_{m,j}}(Y) &=& E_{\snbf{\Pi}_{m-1,j}} \left[\sum_{\snbf{c}_1 \in {\cal C}_{m-1, 2j}} \sum_{\snbf{c}_2 \in {\cal C}_{m-1, 2j+1}} Y^{w_H( \snbf{c}_1 \snbf{\Pi}_{m-1,j} + \snbf{c}_2 |\snbf{c}_2)}\right] \\
     &=& \sum_{\snbf{c}_1 \in {\cal C}_{m-1, 2j}} \sum_{\snbf{c}_2 \in {\cal C}_{m-1, 2j+1}} E_{\snbf{\Pi}_{m-1,j}} \left[  Y^{w_H(\snbf{c}_1 \snbf{\Pi}_{m-1,j}  + \snbf{c}_2 |\snbf{c}_2)}\right].
\eean
By Lemma \ref{th:lem1}, we have
\bean
&& A^{{\cal C}_{m,j}}(Y) \\
&& =
 \sum_{\snbf{c}_1 \in {\cal C}_{m-1, 2j}} \sum_{\snbf{c}_2 \in {\cal C}_{m-1, 2j+1}} \sum_{k = \max(0, w_H(\snbf{c}_1)+w_H(\snbf{c}_2) - n)}^{\min(w_H(\snbf{c}_1), w_H(\snbf{c}_2))} \frac{\left(w_H(\snbf{c}_2) \atop k\right)\left(n-w_H(\snbf{c}_2) \atop w_H(\snbf{c}_1) -k\right)}{\left(n \atop w_H(\snbf{c}_1)\right)}Y^{w_H(\snbf{c}_1)+2w_H(\snbf{c}_2)-2k}.
\eean
The average number of codeword combinations of $\nbf{c}_1$ and $\nbf{c}_2$ for $w_H(\nbf{c}_1) = d_1$ and $w_H(\nbf{c}_2) = d_2$ is $A^{{\cal C}_{m-1, 2j}}_{d_1} A^{{\cal C}_{m-1, 2j+1}}_{d_2}$. The proof is completed.
\end{proof}
Similarly, for the average IOWEF, we have the following theorem.
\begin{thm}
\label{thm:aiowef1}
Given the IOWEFs $A^{{\cal C}_{m-1,2j}}(X,Y)$ and $A^{{\cal C}_{m-1,2j+1}}(X,Y)$, the  IOWEF $A^{{\cal C}_{m,j}}(X, Y)$ of the code ${\cal C}_{m,j} ={\cal C}_{m-1,2j} \nbf{\Pi}_{m-1,j} + {\cal C}_{m-1,2j+1}|{\cal C}_{m-1,2j+1}$ averaged over all possible interleavers $\nbf{\Pi}_{m-1,j}$ is 
\bean
     A^{{\cal C}_{m,j}}(X, Y) 
     &=& \sum_{w_1, w_2, d_1, d_2}A^{{\cal C}_{m-1, 2j}}_{w_1, d_1} A^{{\cal C}_{m-1, 2j+1}}_{w_2, d_2} \sum_{k = \max(0, d_1+d_2 - n)}^{\min(d_1,d_2)} \frac{\left(d_2 \atop k\right)\left(n-d_2 \atop d_1-k\right)}{\left(n \atop d_1\right)}X^{w_1+w_2}  Y^{d_1+2d_2-2k} \\
     &\triangleq& F_m(A^{{\cal C}_{m-1,2j}}(X,Y), A^{{\cal C}_{m-1,2j+1}}(X,Y)),
\eean
where $n = 2^{m-1}$.
\end{thm}
Based on Theorem \ref{lem:awef} and Theorem \ref{thm:aiowef1}, we can compute the average WEF and IOWEF using the following recursive equations
\bea
A^{{\cal C}_{m,j}}(Y) &=& H_m(A^{{\cal C}_{m-1, 2j}}(Y), A^{{\cal C}_{m-1, 2j+1}}(Y)), \label{eq:wefr}\\
A^{{\cal C}_{m,j}}(X, Y) &=&  F_m(A^{{\cal C}_{m-1,2j}}(X,Y), A^{{\cal C}_{m-1,2j+1}}(X,Y)), \label{eq:iowefr}
\eea
for $m=1, \ldots, M$ and $j = 0, \ldots, 2^{M-m}-1$, where the initial conditions are given by
\bean
   A^{{\cal C}_{0,j}}(Y) &=& \left\{
   \begin{array}{ll}
   1 + Y, & j \in {\cal A} \\
   1, & j \in {\cal A}^c
   \end{array}   
   \right.,
\eean
and
\bean
   A^{{\cal C}_{0,j}}(X,Y) &=& \left\{
   \begin{array}{ll}
   1 + XY, & j \in {\cal A} \\
   1, & j \in {\cal A}^c
   \end{array}   
   \right.. 
\eean

The interleaver can also be applied for every output of the lower encoder. In this case, the recursive equation becomes
\[
    \bar{\cal C}_{m,j} =\bar{\cal C}_{m-1,2j}  + \bar{\cal C}_{m-1,2j+1}\nbf{\Pi}_{m-1,j}|\bar{\cal C}_{m-1,2j+1}\nbf{\Pi}_{m-1,j},
\]
for $m=1, \ldots, M$ and $j = 0, \ldots, 2^{M-m}-1$. The initial conditions are $\bar{\cal C}_{0,j} = \{0, 1\}$ if $j \in {\cal A}$ and $\bar{\cal C}_{0,j} = \{0\}$ if $j \in {\cal A}^c$. De-interleaving   $\bar{\cal C}_{m-1,2j}  + \bar{\cal C}_{m-1,2j+1}\nbf{\Pi}_{m-1,j}$ and $\bar{\cal C}_{m-1,2j+1}\nbf{\Pi}_{m-1,j}$ by $\nbf{\Pi}_{m-1,j}^{-1}$ does not change the WEF of the code. Therefore, the following code has the same WEF as $\bar{\cal C}_{m,j}$
\[
       \bar{\cal C}_{m-1,2j}\bar{\nbf{\Pi}}_{m-1,j}  + \bar{\cal C}_{m-1,2j+1}|\bar{\cal C}_{m-1,2j+1},
\]
where $\bar{\nbf{\Pi}}_{m-1,j} = \nbf{\Pi}_{m-1,j}^{-1}$. Since the mapping $\nbf{\Pi}_{m-1,j} \mapsto \bar{\nbf{\Pi}}_{m-1,j}$ is bijective, the interleaver $\bar{\nbf{\Pi}}_{m-1,j}$ is uniform if $\nbf{\Pi}_{m-1,j}$ is uniform. Comparing to (\ref{eq:rpr}), as averaged over all interleavers, $\bar{\cal C}_{m,j}$ has the same average WEF as ${\cal C}_{m,j}$. Therefore, it is sufficient to consider only the former option.

\subsection{WEFs of (32, 16) I-Polar Code and Polar Code}
\begin{table}
\caption{WEF(A): ensemble average WEF of the $(32, 16)$ i-polar code, WEF(B): sample average WEF of the (32, 16) i-polar code with 1000 realizations, WEF(C): exact WEF of the (32, 16) polar code \label{tab:wef32}}
\centering
   \begin{tabular} {|c|r|r|r|r|r|} \hline
   Weight &WEF(A) & WEF(B) & WEF(C) &Type 1 & Type 2 \\ 
   $d$    & Coeff. & Coeff. &  Coeff. & Coeff. & Coeff. \\
   \hline
     0    &1.00      &  1.00         &1       &1         &1       \\ \hline
     4    &8.00      &  8.00         &8       &8         &8       \\ \hline 
     8    &476.24    &  476.29       &700     &476       &508     \\ \hline
     10   &1790.05   &  1789.70      &0       &1792      &1536     \\ \hline 
     12   &7230.82   &  7232.06      &13496   &7224      &8120     \\ \hline
     14   &12530.35  &  12527.87     &0       &12544     &10752     \\ \hline
     16   &21463.06  &  21466.16     &37126   &21446     &23686     \\ \hline
     18   &12530.35  &  12527.87     &0       &12544     &10752     \\ \hline
     20   &7230.82   &  7232.06      &13496   &7224      &8120     \\ \hline
     22   &1790.05   &  1789.70      &0       &1792      &1536     \\ \hline 
     24   &476.24    &  476.29       &700     &476       &508     \\ \hline
     28   &8.00      &  8.00         &8       &8         &8     \\ \hline
     32   &1.00      &  1.00         &1       &1         &1     \\ \hline
   \end{tabular}
\end{table}
The WEFs of the (32, 16) i-polar code and polar code are compared. The set ${\cal A}$ which is used for the i-polar code and polar code is given by
\[
  {\cal A} = \{11,    13,    14,    15,    19,    21,    22,    23,    24,    25,    26,    27,    28,    29,    30,    31\}.
\]
This set ${\cal A}$ is obtained by using the bit channel selection algorithm described previously which was proposed in \cite{Chiu2013}. Table \ref{tab:wef32} gives the coefficients of WEFs of all test cases, of which the Hamming weights, denoted as $d$, are listed in the first column and the remaining columns are the coefficients $A_d$ of all test cases. The WEF averaged over the ensemble of i-polar codes is denoted as WEF(A), which is computed based on the recursive equation (\ref{eq:wefr}). Since the code size is small, the WEF of a realization of i-polar codes can be enumerated exhaustively. We take 1000 independent realizations of i-polar codes and compute the WEF of each realization. In Table \ref{tab:wef32}, WEF(B) denotes the sample average of the WEFs over 1000 realizations. The WEF of the polar code is also enumerated exhaustively and is denoted as WEF(C). It can be observed that the sample average WEF(B) is very close to the (analytical) ensemble average WEF(A). Also, among the 1000 realizations, only two types of WEFs are observed, denoted as WEF type 1 and WEF type 2 as given in Table \ref{tab:wef32}. Among the 1000 realizations, the WEF type 1 and WEF type 2 appear 991 times and 9 times, respectively. 
Note that WEF type 1 and WEF type 2 are just two WEFs that occur with higher probability than the others. There are other WEFs, e.g., the WEF of the polar code, with small probabilities that do not appear among 1000 realizations.
The WEF type 1 and WEF type 2 are all close to the WEF(A), which means that, with high probability, any realizations are as good as the ensemble average WEF(A).  
The WEF(C) of the polar code concentrates to a smaller number of Hamming weights, i.e., there are no codewords of Hamming weights 10, 14, 18, and 22 as shown in Table \ref{tab:wef32}. The reason is that the i-polar code contains interleavers which have the effect of spreading the Hamming weights of codewords widely. 

For linear codes, the minimum Hamming weight, denoted as $d_{\min}$, and its multiplicity, denoted as $A_{d_{\min}}$, dominate the performance at high SNRs. Table \ref{tab:wef32} shows that both polar code and i-polar code have the same number of codewords 8 with minimum Hamming weight 4. This means that both codes have the the same error probability at high SNRs with ML decoding. This phenomenon can be observed from the upper bounds and simulated BLER curves of both codes shown in Figure \ref{fig:polar32_16}. 

The parameters $d_{\min}$ and $A_{d_{\min}}$ for the i-polar and polar codes with $N = 512$ and $K$ varying from 32 to 480 are shown in Figure \ref{fig:dminAdmin}. The parameters $d_{\min}$ and $A_{d_{\min}}$ for i-polar codes are obtained through the recursive equation (\ref{eq:wefr}).  Since there are no analytical WEFs for polar codes, we use the SCL decoder with $L = 5 \times 10^5$ to search for the minimum weight codewords as that proposed in \cite{Li2012}. The results show that both codes have the same minimum Hamming weight. However, the parameters ${A}_{d_{\min}}$ of i-polar codes are smaller than or equal to those of polar codes, which means that, for some cases, i-polar codes perform better than polar codes with ML decoding at high SNRs.



\begin{figure}[!t]
\centering
\includegraphics[width=1.0\columnwidth]{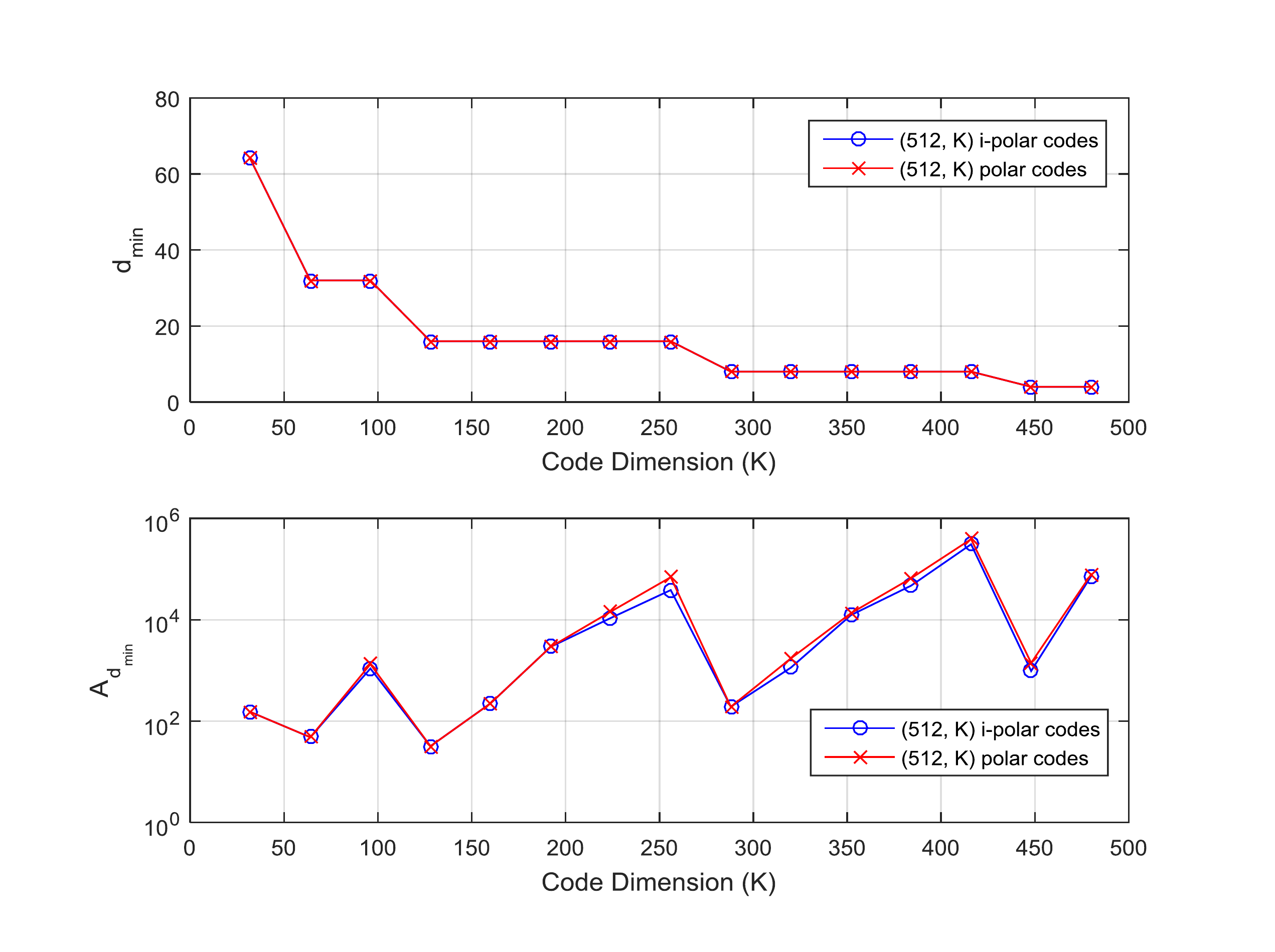}
\caption{Minimum Hamming weights and their multiplicities of $(512, K)$ i-polar  and polar codes with $K$ varying from 32 to 480.} 
\label{fig:dminAdmin}
\end{figure}

\subsection{BLER Upper Bound}
In this paper, we focus on the performance of codes over BI-AWGN channels. For BI-AWGN channels, the $i$th received signal can be represented as
\[
       y_i = \sqrt{E_s}(1-2c_i) + w_i, \ \ \ \ \mbox{for $i=0,1,\ldots, N-1$},
\] 
where $c_i \in \{0,1\}$ is the $i$-th bit of the codeword $\nbf{c} \in {\cal C}$, $w_i$ is the zero-mean additive white Gaussian noise with $E[w_i^2]=N_0/2$, and $E_s$ is the symbol energy.
Given the WEF of a code ${\cal C}$, the union bound on the BLER over BI-AWGN channels is a function of the WEF given by
\beq
    P_{BLER} \leq \sum_{d\neq 0} A^{\cal C}_d Q\left(\sqrt{2d \rho}\right),
\label{eq:PBLER}
\eeq
where $\rho$ is the signal-to-noise ratio defined as $2E_s/N_0$. However, the union bound may be too loose at low SNRs. A tighter upper bound, called {\em simple bound}, was proposed in \cite{Divsalar1999}. For convenience, the bound is given here as
\beq
   P_{BLER} \leq \sum_{d = d_{\mbox{\scriptsize min}}}^{N-K+1} \min\left\{e^{-nE(\rho, \delta)}, A^{\cal C}_d Q\left(\sqrt{2d \rho}\right)\right\},
\label{eq:PBLERU}
\eeq
where $\delta = d/N$, and with $r(\delta) = \ln[A^{\cal C}_d]/N$, 
\[
E(\rho, \delta) = \frac{1}{2} \ln \left[1-2c_0(\delta)f(\rho, \delta)\right] + \frac{\rho f(\rho, \delta)}{1+f(\rho, \delta)}, \ \ c_0(\delta)< \rho < \frac{e^{2r(\delta)}-1}{2\delta(1-\delta)}.
\]
Otherwise,
\[
E(\rho, \delta) = -r(\delta) + \delta \rho.
\]
The functions $c_0(\delta)$ and $f(\rho, \delta)$ are given by
\[
    c_0(\delta) = \left(1 - e^{-2r(\delta)}\right) \frac{1-\delta}{2\delta},
\]
and
\[
   f(\rho, \delta) = \sqrt{\frac{\rho}{c_0(\delta)}+2\rho+\rho^2}-\rho-1.
\]
It should be noted that a similar bound can be employed to obtain the upper bound on the bit error rate (BER) if the IOWEF of the code is known. However, in this paper, we only focus on the BLER upper bound. 

\begin{figure}[!t]
\centering
\includegraphics[width=1.0\columnwidth]{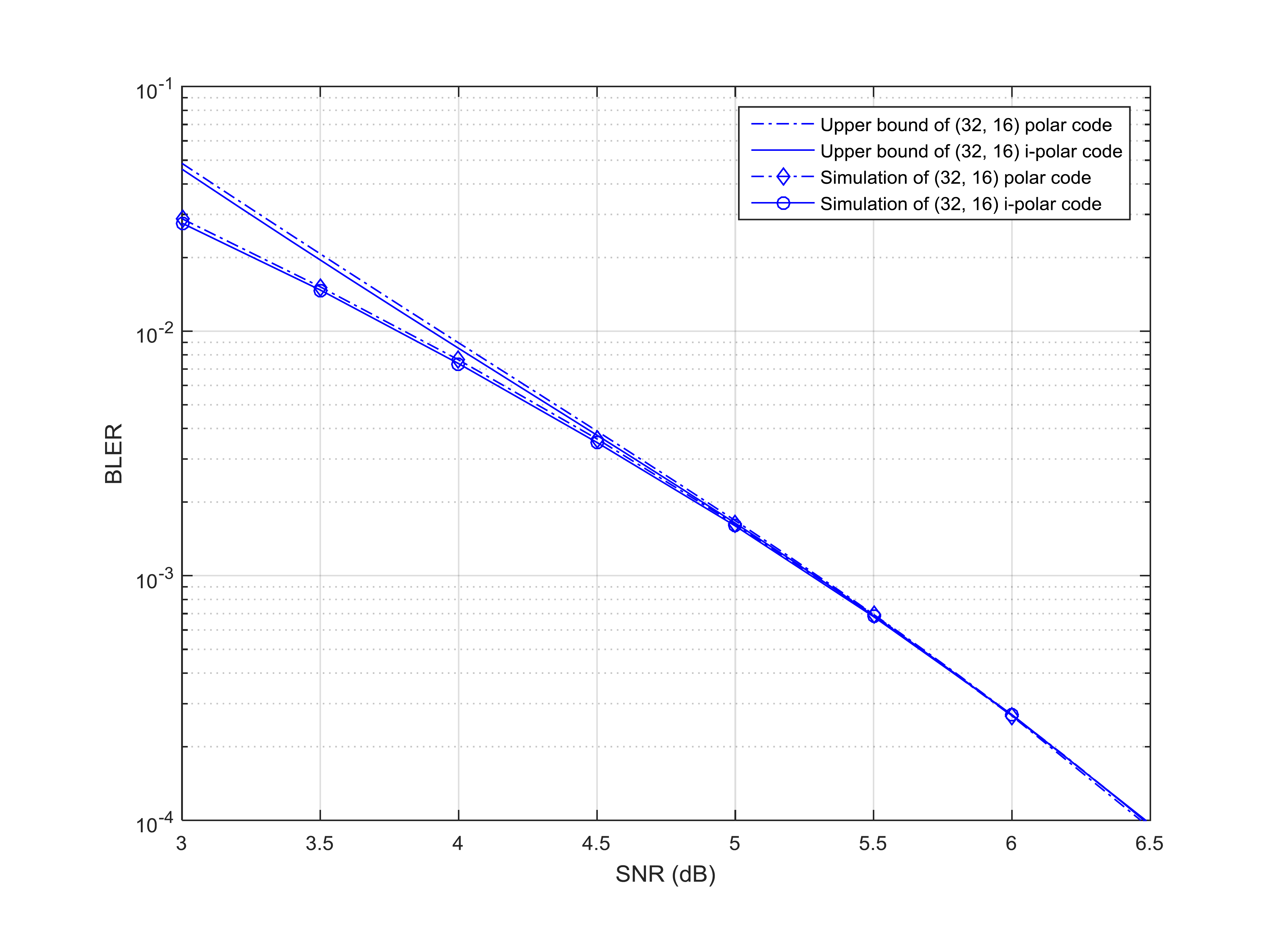}
\caption{BLER upper bounds and simulation results of the $(32, 16)$ polar code and $(32,16)$ i-polar code. 
} 
\label{fig:polar32_16}
\end{figure}

Figure \ref{fig:polar32_16} shows the BLER upper bounds for the $(32, 16)$ polar code and $(32,16)$ i-polar code based on the WEFs given in Table \ref{tab:wef32}. Also BLER simulations are conducted based on the SCL decoder with $L = 8$. In this case, we have verified that the performance of the SCL decoder with $L=8$ is very close to the ML performance. The BLER upper bounds show that the i-polar code is slightly better than the polar code at low SNRs. Simulation results also show such a slight difference.

\section{WEF of Concatenated Coding Schemes}
\label{sec:concat}

\begin{figure} 
\centering
\includegraphics[width=1.0\columnwidth]{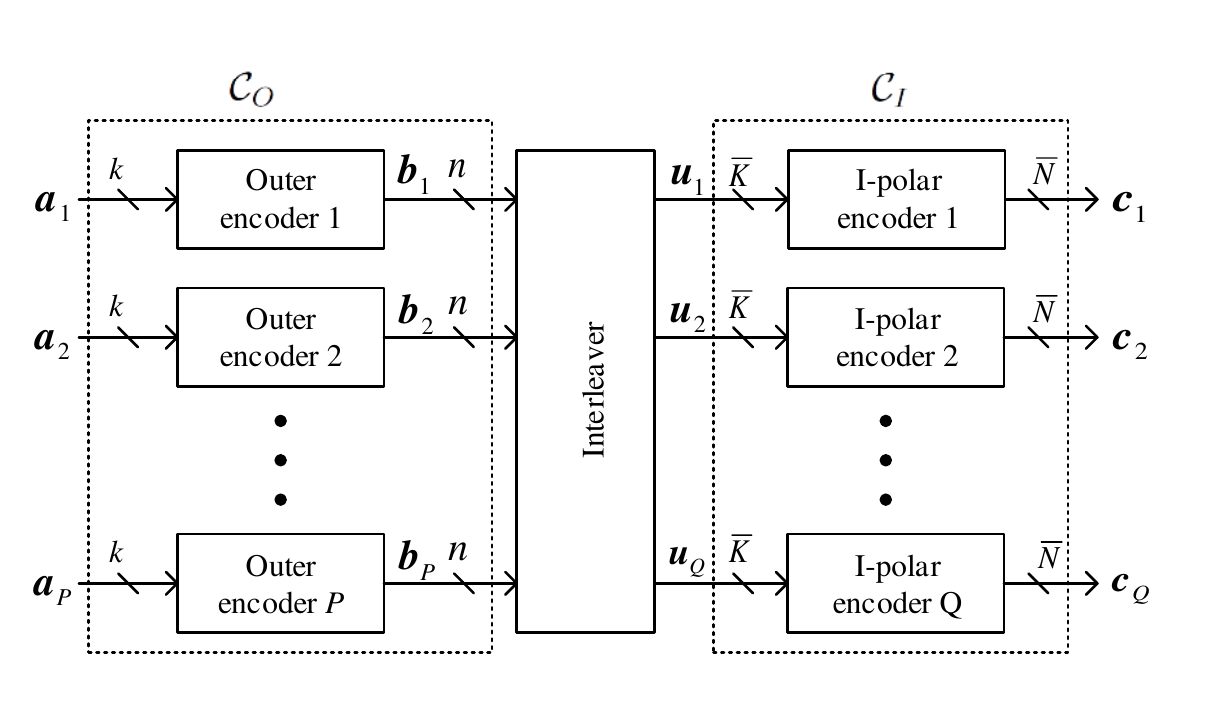}
\caption{Concatenated coding scheme with the i-polar code as the inner component code.} 
\label{fig:conc}
\end{figure}

Polar codes are weak at short to moderate block lengths. Therefore, concatenated coding schemes are often considered in the design of polar codes. A famous coding scheme is to concatenate a CRC code as the outer code with a polar code as the inner code \cite{Tal2015}. The Reed-Solomon codes, BCH codes, convolutional codes, and LDPC codes are also considered as the outer code with a polar code as the inner code \cite{Mahdavifar2014, Wang2016, Guo2014}. So far, for concatenated coding schemes, most research works focus on the asymptotic analysis, i.e., the performance analysis as $N$ approaches infinity. The performance levels for codes of finite block lengths rely on simulations which are time-consuming. We will develop the WEF analysis for concatenated codes which can then be used to evaluate the BLER performance using the bound given in (\ref{eq:PBLERU}).

We consider a concatenated coding scheme as shown in Figure \ref{fig:conc}. The encoder consists of $P$ parallel outer encoders of the same code, denoted as $\bar{\cal C}_O$. The code $\bar{\cal C}_O$ is called the outer component code. For the $p$th outer encoder, the input message word is denoted as a $k$-bit vector $\nbf{a}_p$ and the output codeword is denoted as an $n$-bit vector $\nbf{b}_p$. The $P$ output codewords from the outer encoders form an $nP$-bit super-codeword, denoted as $\nbf{b} = [\nbf{b}_1, \ldots, \nbf{b}_P]$. The $nP$-bit super-codeword is then interleaved by an interleaver which outputs a vector $\nbf{u}$ represented as 
$
     \nbf{u} =  \nbf{b} \nbf{\Pi}
$,
where $\nbf{\Pi}$ is the permutation matrix formed by the interleaver. The vector $\nbf{u}$ is then partitioned into $Q$ blocks with each block of size $\bar{K}$, where $\bar{K}Q = nP$. The partition is given by $\nbf{u} = [\nbf{u}_1, \ldots, \nbf{u}_Q]$ with $\nbf{u}_q$ being a $\bar{K}$-bit vector for $q=1, \ldots, Q$. Then $Q$ parallel i-polar encoders of the $(\bar{N},\bar{K})$ inner component code $\bar{\cal C}_I$ are employed to encode the $Q$ input vectors $\nbf{u}_1, \ldots, \nbf{u}_Q$ and finally the $Q$ output codewords $\nbf{c}_1, \ldots, \nbf{c}_Q$ are obtained. The final super-codeword is $\nbf{c} = [\nbf{c}_1, \ldots, \nbf{c}_Q]$.

As indicated in Figure \ref{fig:conc} by dashed blocks, we may represent the super-code corresponding to the parallel outer encoders as ${\cal C}_O$ and the super-code corresponding to the parallel inner encoders as ${\cal C}_I$. The entire system becomes a simple concatenation of  ${\cal C}_O$ and ${\cal C}_I$ with an interleaver in between. Given the WEF $A^{\bar{\cal C}_O}(Y)$ of the outer component code, and the IOWEF $A^{\bar{\cal C}_I}(X,Y)$ of the inner component code, we can calculate the WEF of the outer super-code ${\cal C}_O$ and the IOWEF of the inner super-code ${\cal C}_I$ by the following theorem.

\begin{thm}
Let ${\cal C}_O$ be the parallel concatenation of $P$ identical outer component codes $\bar{\cal C}_O$ and ${\cal C}_I$ be the parallel concatenation of $Q$ identical inner component codes $\bar{\cal C}_I$.
Given $A^{\bar{\cal C}_O}(Y)$ and $A^{\bar{\cal C}_I}(X,Y)$, the WEF of  ${\cal C}_O$ and IOWEF of ${\cal C}_I$ are
\bea
     A^{{\cal C}_O}(Y) &=& \left[A^{\bar{\cal C}_O}(Y)\right]^P,       \label{eq:ACO}\\
     A^{{\cal C}_I}(X,Y) &=& \left[A^{\bar{\cal C}_I}(X,Y)\right]^Q.  \label{eq:ACI}
\eea
\end{thm}
\begin{proof}
The outer super-code is a parallel concatenation of $P$ identical codes $\bar{\cal C}_O$, i.e., ${\cal C}_O = \{[\nbf{b}_1, \nbf{b}_2,\cdots,\nbf{b}_P]: \nbf{b}_1 \in \bar{\cal C}_O, \nbf{b}_2 \in \bar{\cal C}_O,\ldots, \nbf{b}_P \in \bar{\cal C}_O\}$. The WEF of the outer super-code can be represented as
\bean
A^{{\cal C}_O}(Y) &=& \sum_{\snbf{b}_1 \in \bar{\cal C}_O} \sum_{\snbf{b}_2 \in \bar{\cal C}_O} \cdots \sum_{\snbf{b}_P \in \bar{\cal C}_O}
Y^{w_H([\snbf{b}_1,\snbf{b}_2,\cdots,\snbf{b}_P])}\\
&=& \sum_{\snbf{b}_1 \in \bar{\cal C}_O} \sum_{\snbf{b}_2 \in \bar{\cal C}_O} \cdots \sum_{\snbf{b}_P \in \bar{\cal C}_O}
Y^{w_H(\snbf{b}_1)} Y^{w_H(\snbf{b}_2)} \cdots Y^{w_H(\snbf{b}_P)} \\
&=&  \sum_{\snbf{b}_1 \in \bar{\cal C}_O} Y^{w_H(\snbf{b}_1)}  \sum_{\snbf{b}_2 \in \bar{\cal C}_O} Y^{w_H(\snbf{b}_2)} \cdots  \sum_{\snbf{b}_P \in \bar{\cal C}_O} Y^{w_H(\snbf{b}_P)} \\
&=&  \left[A^{\bar{\cal C}_O}(Y)\right]^P.  
\eean
Similar proof can be used to prove (\ref{eq:ACI}).
\end{proof}

 Similar to that shown in  \cite{Benedetto1998}, by the assumption of uniform interleaver, the WEF of the serial concatenated code, denoted as ${\cal C}$, can be represented as
\bea
A^{\cal C}(Y) &=& \sum_{w, d} \frac{A^{{\cal C}_O}_w A^{{\cal C}_I}_{w, d}}{\left(nP \atop w\right)} Y^d, 
\label{eq:ACY}
\eea
where $A^{{\cal C}_O}_w$ and $A^{{\cal C}_I}_{w,d}$ are the coefficients of (\ref{eq:ACO}) and (\ref{eq:ACI}), respectively.

The CRC-$m$ code is often considered as the outer code which adds $m$ redundant (parity) bits at the end of the message vector and can be treated as a $(k+m, k)$ linear code. However, the WEFs of general CRC codes are not available in the literature.
In order to evaluate the WEF of the concatenated code, we employ the primitive $(2^m-1, 2^m-m-1)$ BCH code (which is actually a Hamming code) as the outer component code in the analysis. The encoder structure of the $(2^m-1, 2^m-m-1)$ BCH code is similar to that of the CRC-$m$ code, except that the generator polynomial of the $(2^m-1, 2^m-m-1)$ BCH code must be a primitive polynomial of degree $m$ and the code length must be $2^m-1$.  Some WEFs of BCH codes can be found in \cite{bchweb}.   

An alternative way is to use  systematic regular repeat-accumulate (RRA) codes or irregular repeat-accumulate (IRA)  codes \cite{Jin2000} as the outer component code. For simplicity, we employ systematic RRA codes in our analysis. The RRA encoder repeats every message bit $d_v$ times and then feeds the entire repeated message bits into an interleaver of size $K d_v$. An accumulate code is then employed to encode the interleaved $K d_v$ bits. The coded bits are then punctured regularly, i.e., only the last bit is retained for every $d_c$ output bits, and a total of $m$ parity bits are retained with $m d_c  = K d_v$. The details of computing the WEFs of systematic RRA codes and IRA codes can be found in \cite{Zhang2007}. 

An issue of the concatenated coding scheme is how to design the decoder that  achieves the ML performance. A possible method is to use the belief propagation (BP) decoder for the inner code which outputs soft information for decoding of the outer code. However, the performance of the BP decoder for a polar code is much inferier to that of the ML decoder. Another method for $Q=1$ is to use the SCL decoder for the inner code and an error detector for the outer code as that proposed for decoding of CRC-aided polar codes \cite{Tal2015}. In this method, the SCL decoder outputs $L$ message vectors in a descending order of their reliabilities. The first message vector that passes the error detection of the outer code is selected as the decoded message vector. If no message vector passes the error detection, the message vector of the largest reliability is selected as the decoded message vector. However, for $Q > 1$, we adopt $Q$ SCL decoders working in parallel and each decoder outputs $L$ candidates of message vectors. Therefore,  $L^Q$ combined message vectors and reliabilities are obtained by considering all possible combinations of message vectors and reliabilities from all decoders. The $L^Q$ combined message vectors are then fed to the outer decoder in a descending order of their reliabilities. The first combined message vector that passes the error detection of the outer code is selected as the output. In the simulation, we  adopt the SCL decoder for the inner code and an error detector for the outer code. 

The advantage for the proposed concatenated code with $Q>1$ is that $Q$ SCL decoders can operate in parallel. This parallel decoding structure can reduce the decoding latency as compared to the SCL decoder for a single polar code of the same length and code rate. For example, for  $Q=2$, if the total block length is $N$, the block length of each i-polar code is $N/2$. These two i-polar codes can be decoded simultaneously which means that two SCL decoders of length $N/2$ can operate in parallel. 
Because the codeword length becomes $N/2$ for each decoder, the decoding latency can be reduced as compared to the SCL decoder for a single polar code of length $N$.
\begin{figure}
\centering
\includegraphics[width=1.0\columnwidth]{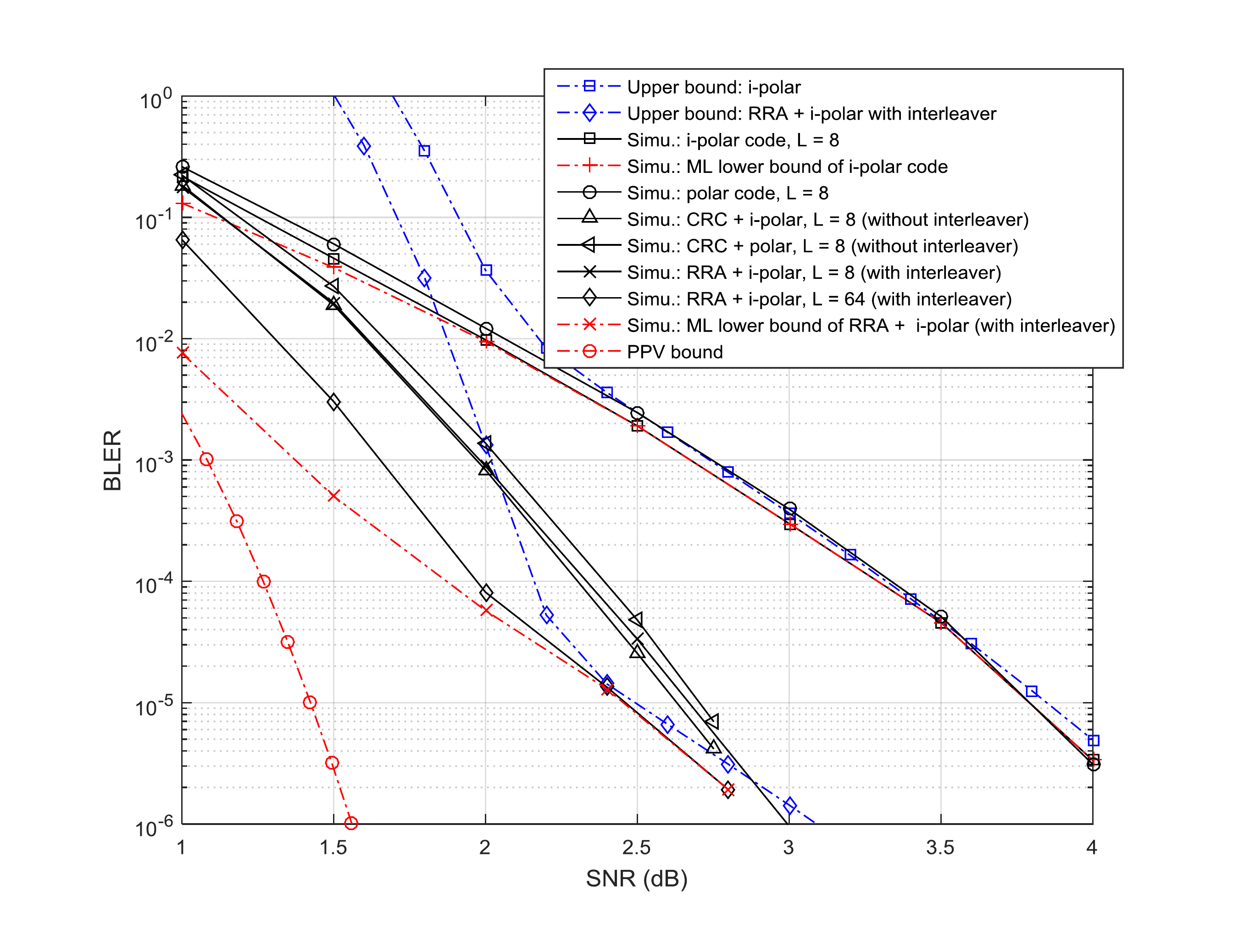}
\caption{Upper bounds and simulation results of (1024, 512) codes.} 
\label{fig:performance_1024_512}
\end{figure}

\section{Analytical and Simulation Results}
\label{sec:simu}
\subsection{Analytical and Simulation Results of I-Polar Codes and Polar Codes}
In order to compare different coding schemes, we consider various rate-1/2 codes of length $1024$ as follows.
\begin{enumerate}
\item  $(1024, 512)$ polar code,
\item  $(1024, 512)$ i-polar code,
\item  CRC-aided polar code which is the direct concatenation of a $(520, 512)$ CRC-8 code and a $(1024, 520)$ polar code without interleaver in between,
\item  CRC-aided i-polar code which is the direct concatenation of a $(520, 512)$ CRC-8 code and a $(1024, 520)$ i-polar code without interleaver in between,
\item RRA-aided i-polar code which is the concatenation of a $(520, 512)$ RRA code with $d_v=3$ and a $(1024, 520)$ i-polar code with an interleaver in between as that shown in Figure \ref{fig:conc} with $P=Q=1$.
\end{enumerate}
The CRC-8 code is generated by the polynomial 
\beq
g_{8A}(D) = D^8 + D^7 + D^6 + D^5 + D^4 + D^3 + 1. 
\label{eq:g8A}
\eeq 
Note that the interleavers, including the one between the inner and outer codes and those in the i-polar code, are only generated once and are fixed for subsequent simulation runs. The SCL decoder with $L=8$ is employed to decode both i-polar and polar codes. A simple error detector is employed to decode the outer code. Figure \ref{fig:performance_1024_512} shows the simulation results and BLER upper bounds based on the simple bound \cite{Divsalar1999} as given in (\ref{eq:PBLERU}). Since the IOWEF of the polar code and the WEF of the CRC code are not available, only upper bounds for the i-polar code and RRA-aided i-polar code are given. 
Simulation results show that the i-polar code is about 0.1 dB better than the polar code at low SNRs. The BLER upper bound of the i-polar code can well predict the performance level at high SNRs. We also give the simulation result of the ML lower bound of the i-polar code as that proposed in \cite{Tal2015}. The ML lower bound is obtained as follows. During the  simulations, when a decoding error occurs, the decoded codeword is checked whether it is more likely than the transmitted codeword. If this event happens, the ML decoder would surely make an error as well. The ML lower bound is the frequency of such an event and is thus a lower bound on the BLER of the ML decoder. Simulation results show that the analytical BLER upper bound and the simulated BLER are very close to the ML performance at high SNRs. 

The performance gap between the CRC-aided i-polar code and CRC-aided polar code is about 0.1 dB at BLER of $10^{-5}$ as shown in Figure \ref{fig:performance_1024_512}. The results also show that the RRA-aided i-polar code is about 0.05 dB worse than the CRC-aided i-polar code at BLER of $10^{-5}$. The RRA-aided i-polar code with SCL decoder of list size $L=8$ is much inferior to that predicted by the BLER upper bound at middle SNRs. This means that $L=8$ is not enough to achieve the ML performance of the RRA-aided i-polar code. To verify this point, we increase the list size to $L=64$ and perform both BLER and ML lower bound simulations. The results show that with $L=64$ the simulated BLER achieves that of the ML lower bound. Also the analytical BLER upper bound can well predict the ML performance at high SNRs. In order to know how far the constructed codes are away to the theoretical limit, we employed the Polyanskyi-Poor-Verd\`{u} (PPV) bound as given in \cite{Polyanskiy2010, Erseghe2016, ppvweb} as a reference. The PPV bound is a lower bound on the size of a code that can be guaranteed to exist with given arbitrary block length $N$, BLER, and SNR of the BI-AWGN channel. Therefore, given block length $N$, BLER, and code dimension $K$, we may find the SNR for which the lower bound on the size of a code equals $2^K$. The BLER curve obtained in this way serves as the BLER lower bound for an $(N, K)$ code. For $N=1024$ and $K=512$, the results shown in Figure \ref{fig:performance_1024_512} indicate that the RRA-aided i-polar code under the SCL decoder with $L=64$ is about 1.0 dB away from the PPV bound at BLER of $10^{-5}$.

\begin{figure}
\centering
\includegraphics[width=1.0\columnwidth]{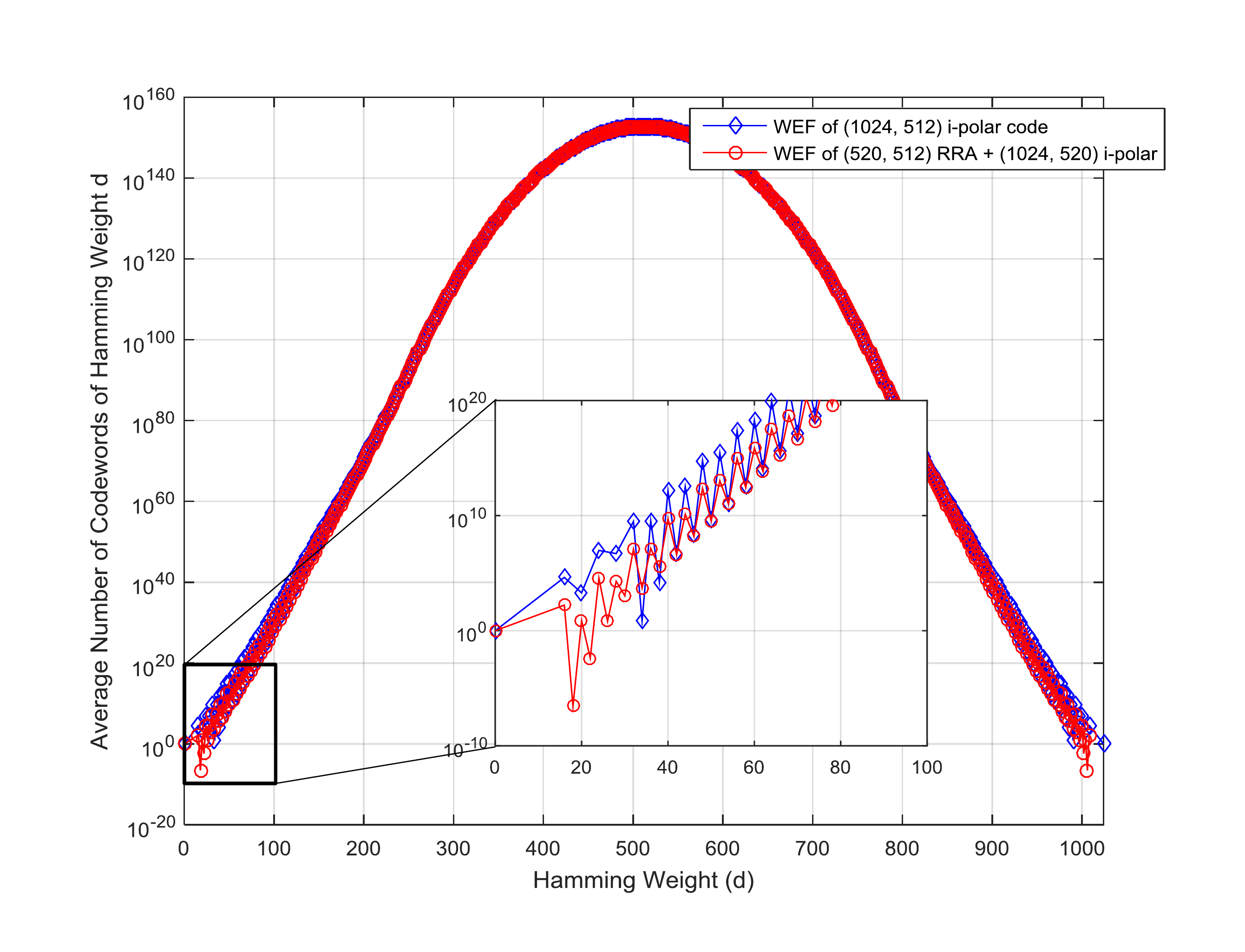}
\caption{WEFs of (1024, 512) codes.} 
\label{fig:WEF_1024_512}
\end{figure}

To understand how the RRA-aided i-polar code performs much better than the i-polar code, we plot the average WEFs of both codes as given in Figure \ref{fig:WEF_1024_512}. 
The average WEF of the (1024, 512) i-polar code is obtained by the recursive equation (\ref{eq:wefr}). The average WEF of the concatenated code is obtained as follows. The IOWEF of the (1024,520) i-polar code is computed by the recursive equation (\ref{eq:iowefr}) and the WEF of the RRA code can be computed as given in \cite{Zhang2007}. Finally, the average WEF of the concatenated code is calculated by (\ref{eq:ACY}).
The results show that both codes have minimum Hamming weight of 16. However, the numbers of low weight codewords of the RRA-aided i-polar code are dramatically reduced as compared to those of the i-polar code. For example, the numbers of the minimum weight codeword of the RRA-aided i-polar code and i-polar code are  $166.84$ and $42403.31$, respectively. This means that, at high SNRs, the BLER of the RRA-aided i-polar code is about  $166.84/42403.31 \approx 3.9\times10^{-3}$ times the BLER of the i-polar code. The upper bounds and simulation results for both codes faithfully reflect the difference as shown in Figure \ref{fig:performance_1024_512}.

\begin{figure}
\centering
\includegraphics[width=1.0\columnwidth]{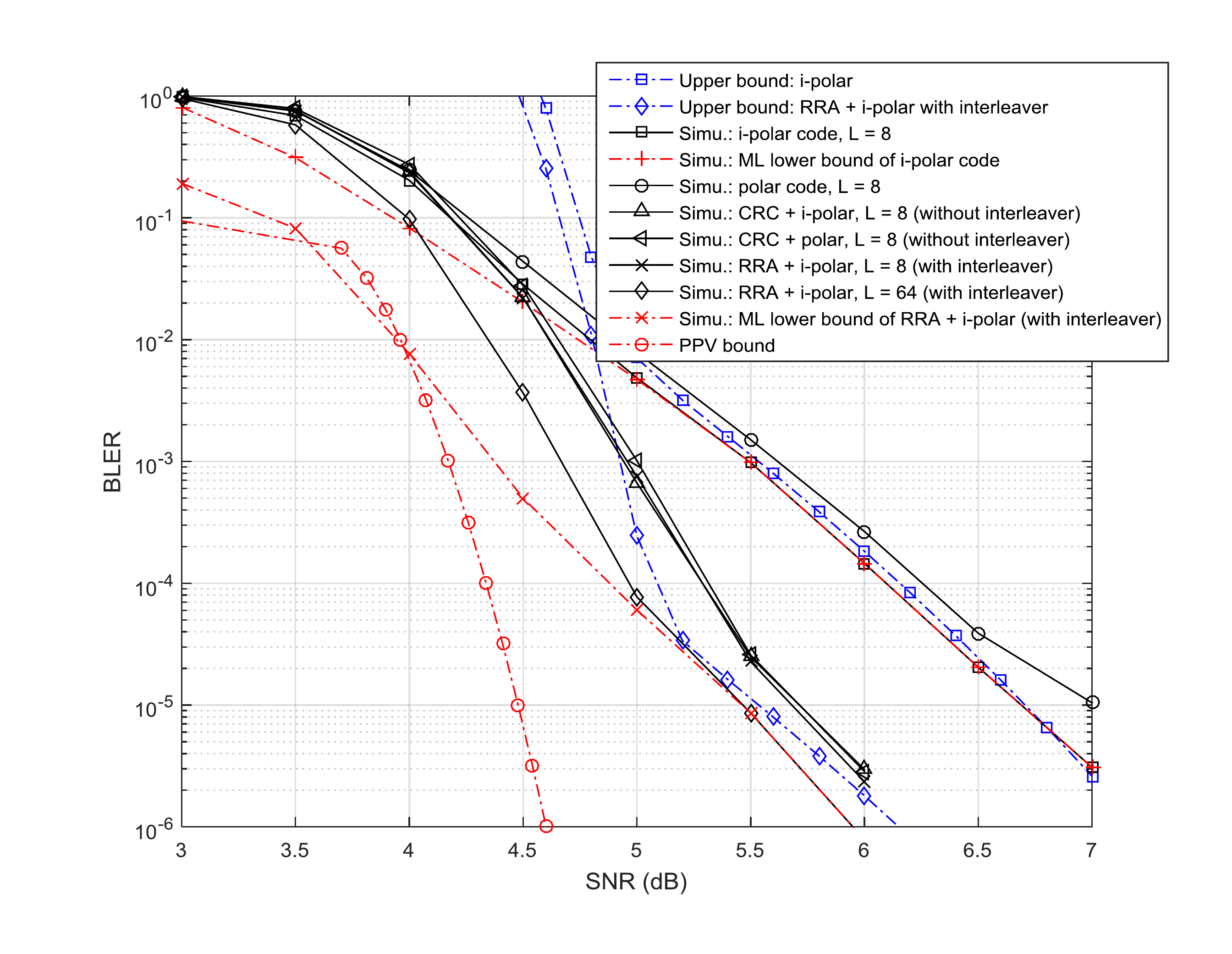}
\caption{Upper bounds and simulation results of (1024, 768) codes.} 
\label{fig:performance_1024_768}
\end{figure}

We also conduct similar simulations by increasing the code rate to 3/4. The CRC-8 code is generated by the polynomial $g_{8B}(D) = D^8 +D^7 +D^6 +D^4 +D^2+1$. The BLER simulation results and upper bounds are given in Figure \ref{fig:performance_1024_768}. The performance gap between the i-polar code and polar code is enlarged to 0.3 dB at BLER of $10^{-5}$. The upper bounds well predict the performance levels at high SNRs.
The results shown in Figure \ref{fig:performance_1024_768} also indicate that the RRA-aided i-polar code under the SCL decoder with $L=64$ is about 1.0 dB away from the PPV bound at BLER of $10^{-5}$.

\begin{figure}
\centering
\includegraphics[width=1.0\columnwidth]{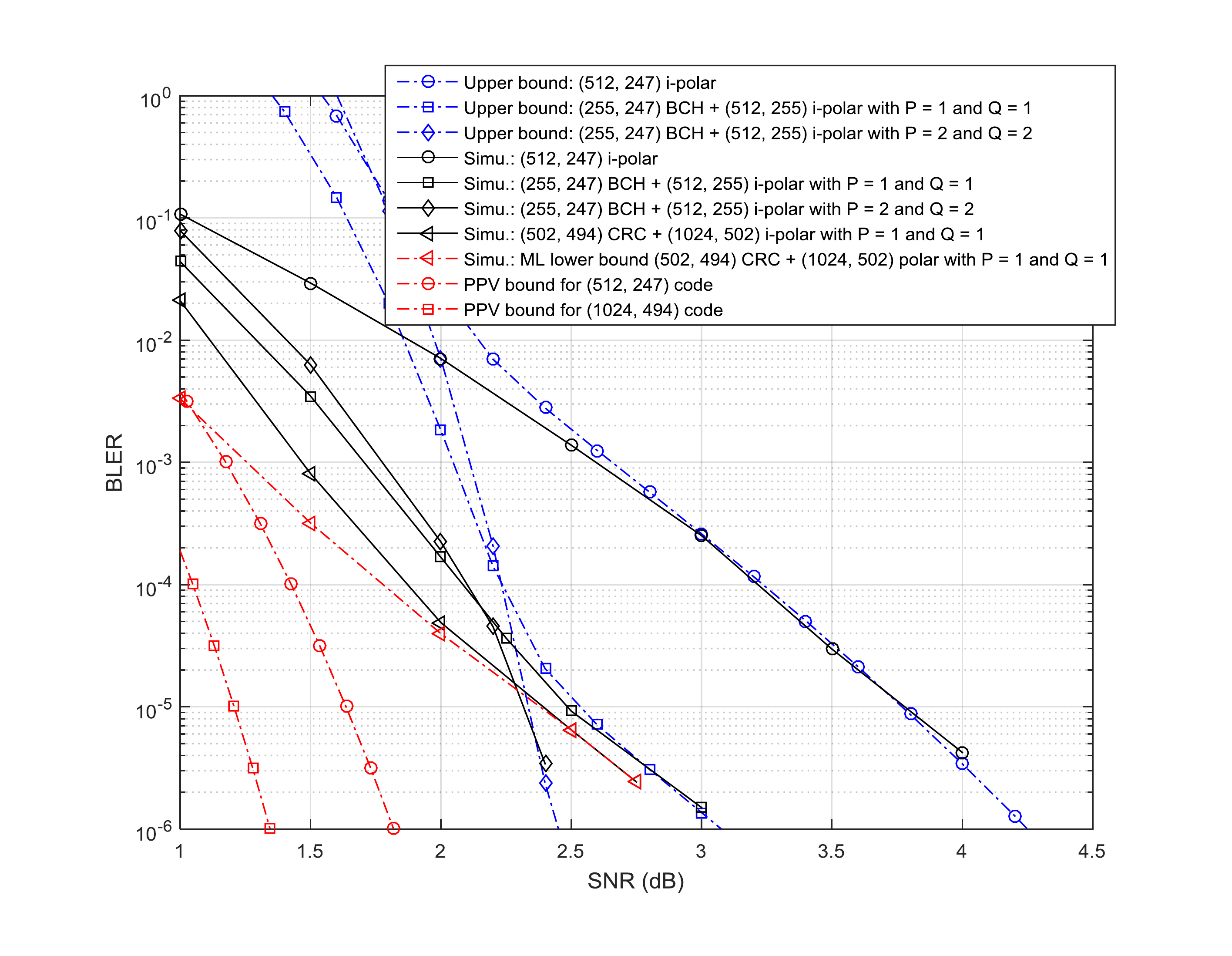}
\caption{Performance of BCH-aided i-polar codes under the SCL decoder with $L=64$.} 
\label{fig:concPerformance}
\end{figure} 

\subsection{BCH-Aided I-Polar Code}
We consider the primitive $(255, 247)$ BCH code as the outer component code and the $(512, 255)$ i-polar code as the inner component code. The generator polynomial of the BCH code is $g(D) = D^8 + D^4+D^3+D^2+1$. The WEF of the $(255, 247)$ BCH code can be found in \cite{bchweb}.
Two settings are considered in the simulation. The first setting is $P=Q=1$ and the second setting is $P=Q=2$. In order to achieve the ML performance, the SCL decoder with $L=64$ is employed for all codes considered.    
The analytical and simulation results are given in Figure \ref{fig:concPerformance}. The results show that the BLER upper bounds can well predict the performance levels at high SNRs. With $P=Q=1$, the BCH-aided i-polar code is about 1.3 dB better than the i-polar code at the BLER of $10^{-5}$. The slope of the BLER curve for $P=Q=2$ is much steeper than that of $P=Q=1$ at high SNRs. This phenomenon is well predicted by the BLER upper bound.

For $P=Q=2$, the concatenated code is a $(1024, 494)$ code. To compare to the code of the same length and code rate with $P=Q=1$, we conduct a simulation of the CRC-aided i-polar code which is the concatenation of the $(502, 494)$ CRC-8 code, generated by $g_{8A}(D)$ as given in (\ref{eq:g8A}), and the $(1024, 502)$ i-polar code with an interleaver in between. The SCL decoder with $L=64$ is employed to decode the CRC-aided i-polar code with $P=Q=1$. The BLER lower bound is also conducted in this simulation. Simulation results shown in Figure \ref{fig:concPerformance} indicate that the proposed $P=Q=2$ coding scheme is worse than the $P=Q=1$ case at low SNRs.  However, the BLER  for the $P=Q=2$ case becomes better than that of the $P=Q=1$ case at high SNRs. The simulation results also show that the BLER  for the $P=Q=2$ case is   better than the ML lower bound for the $P=Q=1$ case at high SNRs. This means that, at high SNRs, the BLER of the $P=Q=2$ case decoded by the SCL decoder with $L=64$ is truly better than that of the $P=Q=1$ case even under ML decoding.

The slope of the BLER curve for the $P=Q=2$ case is steeper than that of the $P=Q=1$ case under the same length and code rate. This means that the proposed coding scheme with $P=Q=2$ can provide more reliable communications at high SNRs. In addition, the proposed coding scheme with $P=Q=2$ allows two SCL decoders to operate in parallel which can reduced the decoding latency. Therefore, in terms of the realiability and decoding latency, the proposed coding scheme is suitable for  ultra-reliable low-latency communications (URLLC) \cite{ITU-R2017}.      

\begin{figure}
\centering
\includegraphics[width=1.0\columnwidth]{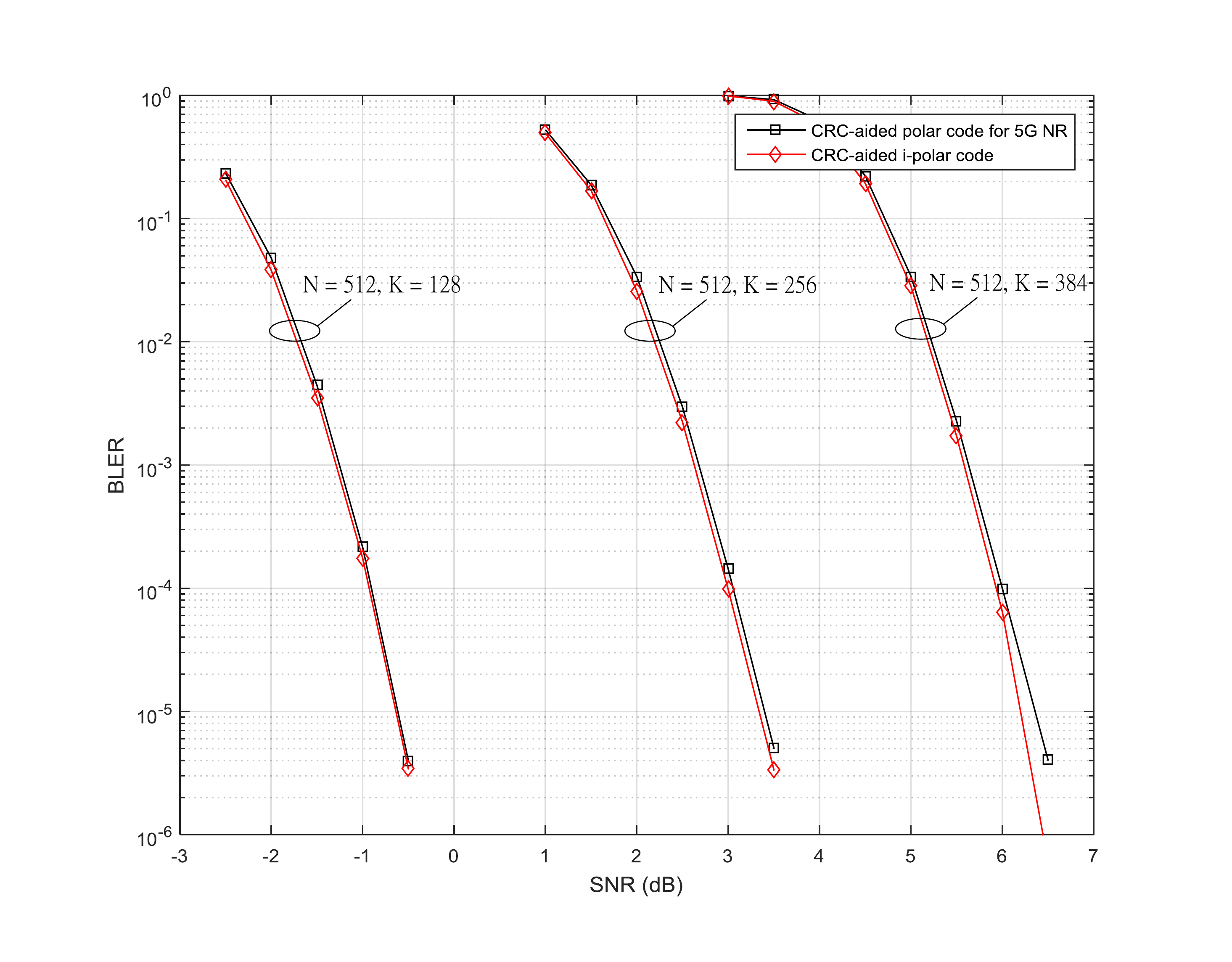}
\caption{Simulation results for 5G polar codes and i-polar codes under the SCL decoder with $L=8$.} 
\label{fig:polar_5G}
\end{figure}

\subsection{I-polar Code and Polar Code for Fifth Generation New Radio (5G NR)}
Polar codes have been employed as the channel coding technique for the control channel in the specification of fifth generation new radio (5G NR) \cite{TS38.212}. The maximum length of polar codes defined in the specification \cite{TS38.212} is $N=1024$. Let ${\cal A}(k)$ be the unfrozen bit set with $k$ elements, i.e., $|{\cal A}(k)|=k$. The 5G polar codes have the property that ${\cal A}(k) \subset {\cal A}(l)$ if $k < l$. Therefore, a single sequence, termed {\em polar sequence}, is defined in Table 5.3.1.2-1 of \cite{TS38.212}, which can be used to derive ${\cal A}(k)$ for any $0 < k \leq N$ and $N\in\{128, 256, 512, 1024\}$. In addition, the CRC-24 with generator polynomial $g_{\mbox{\scriptsize CRC24C}}(D)$ defined in Section 5.1 of \cite{TS38.212} is employed as the outer code. We select three codes of $N=512$ and $K\in\{128,256,384\}$ which correspond to code rates of $1/4$, $1/2$, and $3/4$. In order to test the performance of i-polar codes, we replace the 5G polar code by the i-polar code with the same unfrozen bit set derived from  \cite{TS38.212}. The BLER simulation results under the SCL decoder with $L=8$ are shown in Figure \ref{fig:polar_5G}. The results show that the i-polar code performs slightly better than the 5G polar code for all cases considered. The maximum improvement is about 0.15dB for the rate-3/4 code at BLER of $10^{-5}$. 


\section{Conclusion}
\label{sec:conc}
\subsection{Suggestion for Further Research}
\begin{enumerate} 
\item {\bf Selection of the unfrozen bit set}: Conventional algorithms for the selection of the unfrozen bit set ${\cal A}$ are based on the Gaussian approximation (GA), which is nearly optimal for the SC decoder. However, for a more sophisticated decoder, such as the SCL decoder, the conventional selection algorithms are not optimal. The WEF of i-polar codes can well predict the BLER performance under the ML decoder, which can be used for the selection of the unfrozen bit set so that the performance can be improved under the SCL decoder. A possible method is to use the greedy algorithm as described below. Let ${\cal A}(k)$ be a unfrozen bit set with $|{\cal A}(k)| = k$ for $k=1, \ldots, N$. Then we can use the constraint ${\cal A}(k) \subset {\cal A}(l)$ for $k < l$ in the bit channel selection algorithm. This constraint has also been used in the code design of 5G polar codes \cite{TS38.212}. Given ${\cal A}(k)$, we may calculate the WEFs of all test cases of ${\cal A}^{(j)}(k+1) = {\cal A}(k) \cup j$ for some $j \in {\cal A}^c(k)$. The BLER upper bound at a certain SNR for each test case can be calculated based on the WEF. The set with the smallest BLER upper bound is selected as ${\cal A}(k+1)$. The above procedure is repeated until $k=N$.

\item {\bf Design of concatenated codes}: For $Q>1$, there are $L^Q$ combined message vectors and reliabilities if $Q$ parallel SCL decoders are employed to decode the inner component codes. This number may become prohibitively complex if $Q$ or $L$ is large. To reduce the number of candidates of combined message vectors, a local error detection code for each inner component code can be employed. In this way, each decoder only outputs message vectors that pass the local error detection, which could significantly reduce the number of candidates from each decoder.

\item {\bf Optimization of RRA and IRA outer codes}: Simulation results show that the RRA-aided i-polar code performs slightly worse than the CRC-aided i-polar code. Therefore, we believe that there is room to improve the RRA outer code. One possible way is to adopt IRA codes. The reason is that the IRA code has more degrees of freedom to select the degree distributions of information nodes and check nodes. Then the WEF of the concatenated code can be used to design the degree distributions \cite{Jin2000} of information nodes and check nodes of the IRA code. 

\end{enumerate}
\subsection{Concluding Remarks}

This paper proposes a new class of polar codes, termed interleaved polar (i-polar) codes. By the assumption of uniform interleavers, we derive the average WEF over ensemble of codes. The average WEF is then used to calculate the BLER upper bound. Simulation results show that i-polar codes outperform  polar codes. 
Also, the BLER upper bounds can well reflect the BLER performance levels at high SNRs.
We also propose a concatenated coding scheme which employs $P$ identical high-rate codes as the outer code and $Q$ identical i-polar codes as the inner code.  The concatenated coding scheme shows a steeper BLER slope at high SNRs for $P=Q=2$ and its performance is better than that of the CRC-aided i-polar code with $P=Q=1$ of the same length and code rate at high SNRs. However, the proposed coding scheme allows multiple SCL decoders to operate in parallel, which can significantly reduce the decoding latency. Therefore, the proposed coding scheme is suitable for  ultra-reliable low-latency communications (URLLC).

\bibliographystyle{IEEEtran}
\bibliography{IEEEabrv,ref}

\end{document}